\documentclass[journal,draft,onecolumn]{IEEEtran}

\usepackage[final]{graphicx}
\usepackage[reqno]{amsmath}
\usepackage{amssymb}
\usepackage{url}
\usepackage{color}
\usepackage{times}

\newfont{\bbb}{msbm10 scaled 500}

\newfont{\bb}{msbm10 scaled 1100}

\newcommand{\RR}{\mbox{\bb R}}

\newcommand{\EE}{\mbox{\bb E}}

\newcommand{\Prob}{\textrm{Pr}}

\newcommand{\Bc}{{\cal B}}
\newcommand{\Cc}{{\cal C}}

\newcommand{\Ec}{{\cal E}}

\newcommand{\Nc}{{\cal N}}

\newcommand{\Pc}{{\cal P}}

\newcommand{\Rc}{{\cal R}}
\newcommand{\Sc}{{\cal S}}
\newcommand{\Tc}{{\cal T}}
\newcommand{\Uc}{{\cal U}}

\newcommand{\Xc}{{\cal X}}
\newcommand{\Yc}{{\cal Y}}
\newcommand{\Zc}{{\cal Z}}

\newtheorem{theorem}{Theorem}
\newtheorem{proposition}[theorem]{Proposition}
\newtheorem{corollary}[theorem]{Corollary}
\newtheorem{lemma}[theorem]{Lemma}

\definecolor{OXO-emph}{rgb}{0.61,0.00,0.00}

\newcommand{\qed}{\hfill $\blacksquare$}

\title{State Amplification Subject To Masking Constraints}

\author{
O.~Ozan~Koyluoglu, Rajiv Soundararajan, and Sriram Vishwanath
\thanks{
O. Ozan Koyluoglu is with the Department of Electrical and Computer Engineering, The University of Arizona, Tucson, AZ. Email: ozan@email.arizona.edu.
Rajiv Soundararajan is with Qualcomm Research India, Bangalore, India. Email: rajivs@utexas.edu.
Sriram Vishwanath is with the Department of Electrical and Computer Engineering, The University of Texas, Austin, TX. Email: sriram@austin.utexas.edu.}
\thanks{The material in this paper was presented in part at the
Forty-Ninth Annual Allerton Conference on Communication, Control, and Computing,
Monticello, IL in September 2011.}
}

\begin{document}
\maketitle


\begin{abstract}
This paper considers a state dependent broadcast channel with one transmitter, Alice, and two receivers, Bob and Eve. The problem is to effectively convey (``amplify") the channel state sequence to Bob while ``masking" it from Eve. The extent to which the state sequence cannot be masked from Eve is referred to as leakage. This can be viewed as a secrecy problem, where we desire that the channel state itself be minimally leaked to Eve while being communicated to Bob. The paper is aimed at characterizing the trade-off region between amplification and leakage rates for such a system. An achievable coding scheme is presented, wherein the transmitter transmits a partial state information over the channel to facilitate the amplification process. For the case when Bob observes a stronger signal than Eve, the achievable coding scheme is enhanced with secure refinement. Outer bounds on the trade-off region are also derived, and used in characterizing some special case results. In particular, the optimal amplification-leakage rate difference, called as differential amplification capacity, is characterized for the reversely degraded discrete memoryless channel, the degraded binary, and the degraded Gaussian channels. In addition, for the degraded Gaussian model, the extremal corner points of the trade-off region are characterized, and the gap between the outer bound and achievable rate-regions is shown to be less than half a bit for a wide set of channel parameters.
\end{abstract}


\section{Introduction}


\subsection{Problem Statement}

In this paper, we consider a state dependent broadcast channel model with two users, and investigate the question of to what extent the state of the channel can be amplified at the receiver (Bob) and masked from the other receiver (referred to as Eve). The entire channel state sequence is presumed to be known non-causally to the transmitter (Alice). The only manner to affect the state information at Bob and Eve is by the encoding scheme used at the transmitter. For such a system, we aim to characterize the trade-off between the ``amplification"-rate $R_a$ (at which the legitimate pair can operate) and the ``leakage"-rate $R_l$ (to the eavesdropper).

Formally, consider a discrete memoryless channel given by $p(y,z|x,s)$, where $x\in\Xc$ is the channel input, $s\in\Sc$ is the channel state, and $(y,z)\in(\Yc\times\Zc)$ is the channel output. Here, $y$ corresponds to the received channel output at the legitimate receiver (Bob) and $z$ is the output at the eavesdropper (Eve). The channel is memoryless in the sense that
\begin{equation*}
p(Y^n=y^n,Z^n=z^n| X^n=x^n, S^n=s^n)=
\prod\limits_{i=1}^{n}p(y_i,z_i|x_i,s_i),
\end{equation*}
and the state sequence $S^n$ is independent and identically distributed (i.i.d.) according to a probability distribution indicated by $p(s)$. It is assumed that the channel state sequence is non-causally known at the transmitter. The system model is given in Fig~\ref{fig:Model}.

\begin{figure}[t]
    \centering
    \includegraphics[width=0.7\columnwidth]{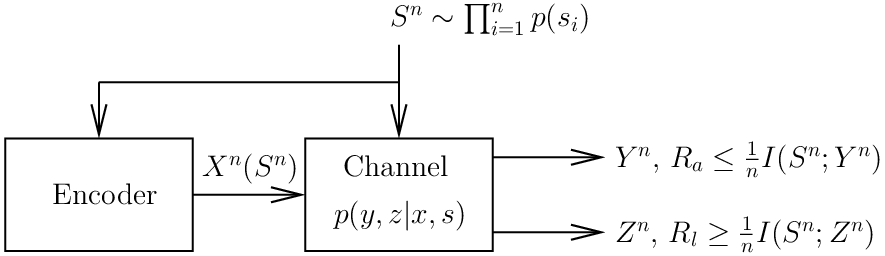}
    \caption{The system model for amplification subject to
    masking problem.}
    \label{fig:Model}
\end{figure}

The task of the encoder is to ``amplify" the state sequence at Bob  (channel output $Y^n$) and to ``mask" the state sequence from Eve ($Z^n$). Formally, the encoder $\textrm{Enc}(S^n,n)$ is a mapping of channel state $S^n$ to the channel input $X^n$, i.e., $\textrm{Enc}: \Sc^n \to \Xc^n$, which can be characterized by a conditional probability distribution $p(x^n|s^n)$. $(R_a,R_l)$ is said to be achievable, if for any given $\epsilon>0$, there exist a sequence of encoders $\{\textrm{Enc}(S^n,n)\}$ such that
\begin{equation}\label{eq:Ra}
\frac{1}{n} I(S^n;Y^n) \geq R_a - \epsilon
\end{equation}
\begin{equation}\label{eq:Rl}
\frac{1}{n} I(S^n;Z^n) \leq R_l + \epsilon
\end{equation}
for sufficiently large $n$, where the mutual information terms are with respect to the joint distribution $p(s^n,x^n,y^n,z^n)=p(x^n|s^n)\prod\limits_{i=1}^n p(s_i)p(y_i,z_i|s_i,x_i)$. The problem is to characterize all achievable ($R_a$,$R_l$) pairs, which we denote by the trade-off region $\Cc$.

The performance of the encoder is also quantified by measuring the difference between the achievable amplification and leakage rates. The \emph{differential amplification rate} $R_d$ is said to be achievable if $R_d=R_a-R_l$ for some $(R_a,R_l)\in\Cc$. The maximum value of the differential amplification rate is called as the \emph{differential amplification capacity}, denoted by $C_d$, where 
\begin{equation*}
C_d=\sup\limits_{(R_a,R_l)\in\Cc} R_a-R_l.
\end{equation*}
Note that this quantity is a property of a given trade-off region $\Cc$ and can play a role for some applications, as this difference measures the knowledge difference between the two receivers regarding the state of the channel. (This quantity, $C_d$, is discussed further in the later parts of the sequel.) In general, we are not only interested in the knowledge difference, but also in the entire rate trade-off region, which constitutes the main focus of this paper.

A cost constraint may also be imposed on the channel input with
\begin{equation}\label{eq:Cost}
\frac{1}{n}\sum\limits_{i=1}^n \EE \{c(X_i)\}\leq C,
\end{equation}
where $c:\Xc\to\RR^+$ defines the cost per input letter and the expectation is over the distribution of the channel input. In this scenario, we say $(R_a,R_l)$ is achievable under the cost function $c(.)$ and at expected cost $C$, if there exists a sequence of encoders satisfying \eqref{eq:Ra}, \eqref{eq:Rl}, and \eqref{eq:Cost} in the limit of large $n$. (We use this constraint for the Gaussian channel, where the cost is the average transmitted power.)

Finally, we note that the equivocation rate (denoted by $\Delta_l$) can be used in this formulation as well. In particular, $\Delta_l$ is said to be achievable, if there exists a sequence of encoders satisfying
\begin{equation*}
\Delta_l \leq \frac{1}{n}H(S^n|Z^n) + \epsilon,
\end{equation*}
for sufficiently large $n$. Accordingly, the achievable $(R_a,\Delta_l)$ pairs can be defined. Hence, the problem can be re-formulated in terms of equivocation rate, where we seek to characterize all achievable $(R_a,\Delta_l)$ pairs in the limit of large $n$. Since both the equivocation and leakage rate notions characterize the same trade-off, both notions can be used interchangeably.


\subsection{Related works and applications}

The problem of communication over state dependent channels is studied by Gel'fand and Pinsker \cite{Gel'fand:Coding80}, where a message has to be reliably transmitted over the channel with non-causal state knowledge at the transmitter. The Gaussian version of the problem is solved in 
\cite{Costa:Writing83} through the famous dirty paper coding scheme. While the wiretap channel is introduced and solved in \cite{Wyner:Wire-tap75}, these results are extended to a broadcast setting in \cite{Csisz'ar:Broadcast78}. The problem of sending secure messages over state dependent wiretap channels is studied in \cite{Chen:Wiretap08,Mitrpant:Achievable06}.

On the other hand, the problems of state amplification and state masking are individually solved
in~\cite{Sutivong:Channel05,Kim:State08,Merhav:Information07} for point-to-point channels. Both \cite{Sutivong:Channel05,Kim:State08} and \cite{Merhav:Information07} consider the problem of reliable transmission of messages in addition to state amplification and state masking, respectively. In this paper, we consider the problem of amplifying the state to a desired receiver while trying to minimize the leakage to (or mask the state from) the eavesdropper. We note that, if we set $R_a=0$ in our problem definition, it reduces to the state masking problem as studied in~\cite{Merhav:Information07}. In other words,
\begin{eqnarray*}
R_a&=&0\\
R_l&=&\min\limits_{p(x|s) \textrm{ s.t. } \EE \{c(X)\}\leq C} I(S;Z)
\end{eqnarray*}
can be shown to be achievable~\cite{Merhav:Information07}.
Also, when $R_l\geq H(S)$, the problem reduces to
a state amplification problem ~\cite{Kim:State08}, and one
can achieve the following rate pair.
\begin{eqnarray*}
R_a&=&\min\{H(S),\max\limits_{p(x|s)} I(X,S;Y)\}\\
R_l&\geq&H(S)
\end{eqnarray*}
These represent two extremes of the trade-off region between the amplification and masking rates.
The main focus of this paper is to characterize the entire trade-off region of amplification and leakage rates. In the following, we list some applications of the proposed model.


\subsubsection{Broadcast channels}

Communication over state dependent channels have a wide set of applications in broadcast channels, especially for the MIMO systems~\cite{Caire:achievable03,Caire:Information06,Jafar:PhantomNet04}. In such multi-user channels, an information-carrying signal can be modeled as a state sequence for another. As the transmitter knows the information-carrying signal for the first user, it can be treated as a known state for the transmission to the second user. In addition, the second signal (intended for the second user) can be considered as an overlay to the first one, especially for multicasting scenarios. (Here, the codeword of the first signal may follow an i.i.d. process, or this can be an uncoded information~\cite{Merhav:Information07}.) In such multi-user settings, the security of the communication arise as an important problem due to the broadcast nature of the medium. Utilizing the framework given in this paper, the amplification and leakage of the signals can be analyzed.


\subsubsection{Cognitive radio and relay channels}

Another relevant setting where our results will be of significant interest is cognitive radio systems~\cite{MitolaIII:Cognitive99,MitolaIII:Cognitive00,Devroye:Achievable06,Goldsmith:Breaking09,Liang:Capacity09,Simeone:Cognitive09,Toher:Secrecy10,Zhang:Secure11}.
For example, in an overlay conginitive radio scenario~\cite{Goldsmith:Breaking09}, the cognitive encoder (Alice) can facilitate the secure communication of the primary signal (the state sequence of the channel) by amplifying the signal at the primary receiver (Bob) while masking it from the eavesdropper (Eve). In such an application, the codebook information of the primary signal may be absent at the cognitive radio, and the formulation we have in this paper would be relevant. More generally, this setting belongs to user cooperation or relaying architectures that can increase the security of the communication systems \cite{Lai:relay-eavesdropper08,Tandon:Secure09,Koyluoglu:Cooperative11,Yuksel:Secure11}. In an another cognitive radio setup, two cognitive users (Alice and Bob) can utilize the primary signal (interfering sequence $S^n$) to share a secret key between each other in the presence of an eavesdropper. (Applications to key sharing from channel states are detailed in the following.)


\subsubsection{Secret key generation from channel states}

In recent years, a bridge between cryptography and information theoretic security has emerged in the form of channel state-dependent key generation \cite{Azimi-Sadjadi:Robust07,Mathur:Radio-telepaty08,Jana:effectiveness09,Mathur2010}. In this line of work, a state dependent channel (such as the wireless channel) is considered, and a function of this state is intended to be a ``shared secret'' between the legitimate transmitter (Alice) and legitimate receiver (Bob) while aiming to keep the eavesdropper (Eve) as much in the dark as possible. In the best case, the state(s) seen by Bob and Eve will be completely different (independent). However, the states of the channels are usually dependent. Moreover, for some models (e.g., the dirty paper coding setup~\cite{Costa:Writing83}), there is a single channel state defining the channel for both Bob and Eve. (For example, $Y=X+S+N_{\textrm{Bob}}$ at Bob and $Z=X+S+N_{\textrm{Eve}}$ at Eve as considered in \cite{Khisti:Secret11}. See also the discussion above for cognitive radio channels.) For such scenarios, as long as there is a non-trivial difference between the amplification and leakage rates, the state knowledge can be used to develop shared keys and enable cryptographic algorithms. For instance, utilizing the coding schemes proposed here, Bob can have $nR_a$ bits of information regarding the channel state, where at most $nR_l$ number of these bits are leaked to Eve. Then, privacy amplification~\cite{Bennett:Generalized95} can be utilized to distill secret keys. That is, the methods provided in this paper can be utilized for the ``advantage distillation'' phase of a key agreement protocol~\cite{Bennett:Generalized95}. In this form, the problem at hand is very much related to the generation of secrecy using sources and channels. Here, not only the source in the model is the channel state (different from the models studied in~\cite{Khisti:Secret-key08,Prabhakaran:Secrecy08}), but also it is non-trivially combined with the encoded signals of Alice (via the state dependent channel) to produce the observations at Bob and Eve. Hence, our formulation can be closely associated with the problem of secret key agreement over state dependent channels~\cite{Khisti:Secret10, Khisti:Secret11}. In such problems, one is interested in the design of coding strategies that allow an agreement of secure bits between the legitimate users utilizing the state dependent channel. For example, sending secure bits over the channel will increase the secret key rate~\cite{Khisti:Secret11, Chen:Wiretap08, Mitrpant:Achievable06}. On the other hand, the problem studied in this work, when specialized to the differential amplification measure, is an analysis of the channel state knowledge difference. This analysis is very much related to the question of how many secret bits can be extracted from a given channel state (\cite{Azimi-Sadjadi:Robust07,Mathur:Radio-telepaty08,Jana:effectiveness09,Mathur2010}). We intend to provide information-theoretic guarantees (achievable rates and upper bounds) to the extent to which such a shared secret can be realized for our system model \footnote{We note that differential amplification capacity is related to the secret key rate that can be achieved utilizing the corresponding channel state. For example, consider $Y=S_1$, $Z=S_2$, where $S_1$ and $S_2$ are uniform bits. Here, $C_d=0$ when the state is defined as $S=[S_1,S_2]$, and $C_d=1$ when the state is taken as $S_1$, implying that $1$ bit of secret key rate can be supported.}. Once obtained, this shared secret can now be employed to seed numerous symmetric key cryptosystems~\cite{Delfs:Introduction07, Goldreich:Foundations04}. 


\subsubsection{Watermarking and channel estimation}

Consider that a host image $S$ is utilized in a watermarking scenario, where the encoder modifies the image in order to amplify it at the receiver $Y$ and mask it from receiver $Z$. This model is similar to channel estimation scenarios in wireless communications. For example, a pilot signal can be constructed at the transmitter such that it not only facilitates the fading gain estimation at the receiver but also hides the channel state from the eavesdroppers as much as possible. For instance, while communicating to a base station for channel state estimation, a mobile user may want to hide this information from the external nodes in order to establish privacy of her location.


\subsection{Summary of results and organization}

In this paper, we aim at developing an understanding of the trade-off between amplification and masking rates in a state dependent broadcast channel through achievable regions and outer bounds, and characterizing special cases when they match. The main results of this paper can be summarized as follows.

\subsubsection{Achievable Regions} The main achievability argument of the paper utilizes the transmission of a state dependent message over the state dependent channel. To facilitate this, we construct a (Gel'fand-Pinsker) codebook, wherein the corresponding codeword (denoted by $U^n$) carry this refinement information in such a way that reliable communication can be achieved over the state-dependent channel. Subsequently, we utilize this refinement information within a Wyner-Ziv coding scheme to derive expressions for the achievable amplification rates. In particular, we show that, even though the side information is not generated in an i.i.d. fashion, Wyner-Ziv approach can be used to facilitate amplification process. The leakage rates are then determined by deriving single-letter bounds on $\frac{1}{n} I(S^n;Z^n)$, and the achievable regions are established over the input probability distributions $p(u,x|s)$. The bounds show that the rate of the refinement message not only increases the amplification rate $R_a$, but also increases the leakage rate $R_l$, thereby establishing a trade-off for implementation.

\subsubsection{Secure Refinement} We show that it is possible to extend the proposed region by transmitting \emph{secure refinement} information when Bob observes a ``stronger" channel than Eve. In precise terms, this corresponds to instances of $p(u,x|s)$ satisfying $I(U;Y)\geq I(U;Z)$. Note that a channel is said to have a less noisy structure if $I(U;Y)\geq I(U;Z)$ for all input probability distributions~\cite{Korner:source75, Korner:Comparison77}. We find that the utilization of the notion of secure refinement approach is critical to such channels, and we show that the leakage due to transmission of the message can be minimized by securing the message. In the process of establishing these results, we also develop an alternate proof of secure message transmission over state dependent wiretap channels.

\subsubsection{Special Classes of Channels} Our outer bound arguments
are based on upper bounding $\frac{1}{n} I(S^n;Y^n)$
and lower bounding $\frac{1}{n} I(S^n;Z^n)$. The
achievable schemes and outer bounds presented are used to establish optimality results
for a class of channels. In particular, we show that the proposed scheme
achieves the optimal differential amplification capacity (i.e., the maximum value of
$R_a-R_l$ over the set of achievable $(R_a,R_l)$
pairs) for the reversely degraded discrete memoryless channel (DMC), the degraded binary channel, and the degraded Gaussian channel.

\subsubsection{Gaussian Channels} We characterize the corner points
of the region for the degraded Gaussian channel. In this scenario, we further
bound the gap between achievable and converse regions, and
show the following: Let us denote the message capacity
of Bob's channel as $C_b=\frac{1}{2}\log(1+\textrm{SNR}_b)$ and that of Eve's channel as $C_e=\frac{1}{2}\log(1+\textrm{SNR}_e)$, where 
$\textrm{SNR}_b$ ($\textrm{SNR}_e$) is the signal-to-noise ratio of
Bob (respectively, Eve).
Then, for any given leakage rate $R_l$, we show that the gap between the upper and lower bounds
on the amplification rate $R_a$
is bounded by $C_b$. Similarly, for any given amplification rate $R_a$, we show that the achievable leakage rate is within $C_e$ of the lower bound on $R_l$.
In particular, the corresponding gaps are within half a bit when
$\textrm{SNR}_b\leq 1$ and $\textrm{SNR}_e\leq 1$, respectively.

The rest of the paper is organized as follows. Section~\ref{sec:InnerBound} details the proposed coding schemes and presents the corresponding achievable regions. Section~\ref{sec:OuterBound} provides the outer bounds to the trade-off region. Section~\ref{sec:DMC} includes optimality discussions and numerical results for special classes of DMCs including the reversely degraded channel, modulo additive binary channel model, and the memory with defective cells model. The Gaussian channel model is considered in Section~\ref{sec:Gauss} along with corresponding optimality results. Finally, we conclude the paper in Section~\ref{sec:Conclusion}. Some of the proofs are collected in appendices to improve the flow of the paper.


\section{Achievable Regions}
\label{sec:InnerBound}

We consider transmitting \emph{state dependent information} over the state dependent channel. This refinement information can be utilized at Bob in resolving some ambiguity regarding the channel state. We divide the discussion into two parts, each providing a different scheme for the transmission of this refinement information.


\subsection{State enhanced messaging: Refinement}

In order to facilitate the transmission of information regarding the state sequence, we consider an encoder that transmits a \emph{state dependent message} over the state dependent channel. Here, by employing the Gel'fand-Pinsker coding~\cite{Gel'fand:Coding80}, a rate of $I(U;Y)-I(U;S)$ can be reliably communicated over the state dependent channel. Utilizing this communication rate, Alice can provide a \emph{refinement} information to Bob. In particular, we consider a Wyner-Ziv coding~\cite{Wyner:rate76} approach to provide such a refinement information (as discussed in~\cite{Kim:State08}). Here, the side information at the receiver is $(U^n,Y^n)$ - consisting of Gel'fand-Pinsker codeword $U^n$ and the observed sequence from the channel $Y^n$. For this side information scenario, we utilize the Gel'fand-Pinsker information rate as the bin index of a covering codeword $V^n$. We note that, in the original setup of source coding with side information (aka Wyner-Ziv model), the side information is generated i.i.d. with the source sequence. For the scenario considered here, the side information is not in this form. In the following, we show that this issue can be resolved, and the Wyner-Ziv coding scheme can be utilized to increase the amplification rate. On the other hand, this transmission scheme may leak some state information to Eve. Accordingly, we obtain appropriate bounds on the leakage rate that depends on the Gel'fand-Pisker coding rate used for the refinement. The corresponding achievable region is given by the following.

\begin{theorem}\label{thm:Rin1}[Refinement]
Let $\Rc^{1}$ be the closure of the union of all $\left( R_a, R_l \right)$ pairs satisfying
\begin{eqnarray}
R_a &\leq& I(S;Y,U) + R_r\nonumber\\
R_l &\geq& \min \big\{I(U,S;Z), I(S;Z,U)+R_r \big\}\nonumber\\
R_r &\leq& \min\{I(U;Y)-I(U;S),H(S|Y,U)\},\nonumber
\end{eqnarray}
over all distributions $p(u,x|s)$ satisfying $I(U;Y)\geq I(U;S)$.
Then, $\Rc^{1}\subseteq\Cc$.
\end{theorem}


\subsubsection{Proof of Theorem~\ref{thm:Rin1}}
We provide main steps of the coding argument here and relegate some of the details to appendices.

\textit{Codebook Generation:} Fix a $p(u|s)$ and a $p(v|s)$. (The requirement on these distributions will be specified later.) Randomly and independently generate $2^{nR_v}$ sequences $V^n(l)$, $l \in [1,2^{nR_v}]$, each according to $\prod_{i=1}^n p_V(v_i)$. Partition indices $l$ into equal-size subsets referred to as bins $\Bc(b)=[(b-1)2^{n(R_v-R_r)}+1: \: b2^{n(R_v-R_r)}]$, $b\in[1:2^{nR_r}]$. (This is the codebook used for Wyner-Ziv coding~\cite{Wyner:rate76,ElGamal:Network11}.) For each $m\in[1:2^{nR_m}]$, generate $2^{n(R_u-R_m)}$ number of $U^n$ sequences randomly and independently according to $\prod_{i=1}^n p_U(u_i)$. Index these sequences as $U^n(m,k)$ with $k\in[1:2^{n(R_u-R_m)}]$. (This is the codebook used for Gel'fand-Pinsker coding~\cite{Gel'fand:Coding80,ElGamal:Network11}.)

\textit{Encoding:} Here, the encoder sets $R_m=R_r$ and $m=b$ and transmits message $b$ of rate $R_r$ with Gel'fand-Pinsker coding in order to perform refinement via Wyner-Ziv coding. Given $s^n$, the encoder finds an index $l$ such that $(s^n,v^n(l))\in\Tc_{\epsilon'}^{(n)}$. If there exists more than one index, it selects one uniformly random among these. If there exists no such index, it selects one uniformly random from $[1,2^{nR_v}]$. From $l$, encoder determines the bin index $b$ in the $V^n$ codebook such that $l\in\Bc(b)$. Then, it sets $m=b$ and finds an index $k$ such that $(s^n,u^{(n)}(m,k))\in\Tc_{\epsilon'}^{(n)}$. If no such (covering) index $k$ exists or if there are more than one, the encoder picks one uniformly at random. The encoder then transmits $x^n$ that is generated i.i.d. according to $\prod_{i=1}^n p_{X|U,S}(x_i|u_i(m,k),s_i)$.

\textit{Amplification rate analysis:} In the following, we show that the decoder can obtain $U^n$ codeword using Gel'fand-Pinsker decoding and then $V^n$ codeword by utilizing the side information $(U^n,Y^n)$. (This is the discussion provided in~\cite{Kim:State08} for a state amplification mechanism, which we detail here.) We then provide a derivation of the state amplification rate utilizing this decoding mechanism.

Let $\epsilon>\epsilon''>\epsilon'$. Upon receiving $y^n$, decoder declares that $\hat{m}\in[1:2^{nR_r}]$ and $\hat{k}\in[1:2^{n(R_u-R_r)}]$ are chosen at the encoder, if these are the unique indices such that $(u^n(\hat{m},\hat{k}),y^n)\in\Tc_\epsilon^{(n)}$, otherwise it declares an error. (We remark that compared to the Gel'fand-Pinsker setup~\cite{Gel'fand:Coding80,ElGamal:Network11}, where only the message $m$ has to be uniquely decoded, we require decoder to obtain the codeword $u^n$ chosen at the encoder. This will be utilized in decoding error analysis.) Then, the decoder finds the unique index $\hat{l}\in \Bc(\hat{m})$ such that $(v^n(\hat{l}),u^n(\hat{m},\hat{k}),y^n)\in\Tc_\epsilon^{(n)}$. Otherwise, it declares an error.

Consider the error event $\Ec=\{\hat{M}\neq M, \hat{K}\neq K, \hat{L}\neq L\}$, where $(K,L,M)$ are the indices chosen at the transmitter, and $(\hat{K},\hat{L},\hat{M})$ are the indices decoded at the receiver. Decoder makes an error only if one or more of the following events occur.
\begin{align*}
&\Ec_1 = \{(U^n(M,k),S^n)\notin \Tc_{\epsilon'}^{(n)} \: \textrm{for all } k\}\\
&\Ec_2 = \{(U^n(M,K),Y^n)\notin \Tc_{\epsilon}^{(n)} \}\\
&\Ec_3 = \{(U^n(m,k),Y^n)\in \Tc_{\epsilon}^{(n)} \: \textrm{for some } m\neq M\}\\
&\Ec_4 = \{(U^n(M,k),Y^n)\in \Tc_{\epsilon}^{(n)} \: \textrm{for some } k\neq K\}\\
&\Ec_5 = \{(V^n(l),S^n)\notin \Tc_{\epsilon'}^{(n)} \: \textrm{for all } l\}\\
&\Ec_6 = \{(V^n(L),S^n,U^n(M,K),Y^n)\notin \Tc_{\epsilon}^{(n)} \}\\
&\Ec_7 = \{(V^n(l),U^n(M,K),Y^n)\in \Tc_{\epsilon}^{(n)} \: \textrm{for some } l\neq L, l\in \Bc(M) \}
\end{align*}
Here, $\Ec=\cup_{j=1}^7 \Ec_j$. By the union of events bound, we have $$\Pr\{\Ec\}\leq 
\Pr\{\Ec_1\}+\Pr\{\Ec_1^c \cap \Ec_2\}+\Pr\{\Ec_3\}+\Pr\{\Ec_3^c \cap \Ec_4\}+\Pr\{\Ec_5\}+\Pr\{\Ec_1^c \cap \Ec_5^c \cap \Ec_6\}+ \Pr\{\Ec_7\}.$$
We bound each term above in the following.

a) By covering lemma~\cite{ElGamal:Network11}, $\Pr\{\Ec_1\}\to 0$ as $n\to\infty$, if
\begin{equation}\label{eq:RuminusRm}
R_u-R_m > I(U;S)+\delta(\epsilon').
\end{equation}

b) $\Ec_1^c$ implies that $(U^n(M,K),S^n)\in \Tc_{\epsilon'}^{(n)}$, which implies that $(U^n(M,K),S^n,X^n)\in \Tc_{\epsilon''}^{(n)}$ due to i.i.d. generation of $x_i$ from $(s_i,u_i)$ and the conditional typicality lemma~\cite{ElGamal:Network11} for some $\epsilon''>\epsilon'$. Similarly, as $Y^n$ is generated i.i.d. from $(s_i,u_i,x_i)$ through $(s_i,x_i)$, we have $\Pr\{\Ec_1^c \cap \Ec_2\}\to 0$ as $n\to \infty$ for some $\epsilon>\epsilon''$, again due to the conditional typicality lemma~\cite{ElGamal:Network11}.

c) Each $U^n(m,k)$ for some $m\neq M$ is distributed i.i.d. $\prod\limits_{i=1}^n p_U(u_i)$ and independent of $Y^n$. By packing lemma~\cite{ElGamal:Network11}, we have $\Pr\{\Ec_3\}\to 0$ as $n\to\infty$, if 
\begin{equation}\label{eq:Ru}
R_u < I(U;Y)-\delta(\epsilon).
\end{equation}
The three error analyses above are the same arguments used for decoding the message index in Gel'fand-Pinsker coding~\cite{Gel'fand:Coding80,ElGamal:Network11}.

d) The analysis for showing that $\Pr\{\Ec_3^c\cap\Ec_4\}\to 0$ as $n\to\infty$ is detailed in Appendix~\ref{App:CoveringIndexDecoding}. We use the arguments given in~\cite{Lim:Lossy10} in order to show this. The original problem studied in~\cite{Lim:Lossy10} does not involve a state-dependent channel, but the coding scheme constructs channel inputs via $x(u,s)$, which can be viewed as a state dependent channel. (The argument we have here - the Gel'fand-Pinsker codeword decoded at the decoder is the same as the one chosen at the encoder - also appears in~\cite{Choudhuri:non12}, which states that this observation follows from the arguments in~\cite{Lim:Lossy10,Lapidoth:Sending10}.)

e) By covering lemma~\cite{ElGamal:Network11}, $\Pr\{\Ec_5\}\to 0$ as $n\to\infty$, if
\begin{equation}\label{eq:Rv}
R_v > I(S;V)+\delta(\epsilon').
\end{equation}

f) For analyzing $\Pr\{\Ec_1^c \cap \Ec_5^c \cap \Ec_6\}$, we note that one may not utilize conditional typicality lemma as done in the proof of Wyner-Ziv coding (see, e.g., \cite[Section 11.3.1]{ElGamal:Network11}). Because, the side information here, $(U^n,Y^n)$, is not generated i.i.d. through $p_{U,Y|S}(u_i,y_i|s_i)$. Therefore, one can not go from joint typicality of $(V^n(L),S^n)$ to the joint typicality of $(V^n(L),S^n,U^n(M,K),Y^n)$. Instead, we consider the following approach: $\Ec_5^c$ implies that 
\begin{align}\label{eq:f1}
(V^n(L),S^n)\in\Tc_{\epsilon'}^{(n)}, 
\end{align}
and $\Ec_1^c$ implies that $(U^n(M,K),S^n)\in\Tc_{\epsilon'}^{(n)}$. Now, consider a pair $(v^n,s^n)\in\Tc_{\epsilon'}^{(n)}$. In addition, we have $\Pr\{U^n(M,K)=u^n|S^n=s^n,V^n(L)=v^n\}=\Pr\{U^n(M,K)=u^n|S^n=s^n\}$ such that the pmf $p(u^n|s^n)$ satisfies the following two conditions: The first condition is that
\begin{align}\label{eq:f2}
\lim\limits_{n\to\infty} \Pr\{(s^n,U^n(M,K))\in\Tc_{\epsilon'}^{(n)}(S,U)\}=1, 
\end{align}
which is due to the fact that probability of the $\Ec_1^c$ event vanishes. And, the second condition is that, for every $u^n\in\Tc_{\epsilon'}^{(n)}(U|s^n)$ and for sufficiently large $n$,
\begin{align}\label{eq:f3}
2^{-n(H(U|S)+\delta(\epsilon'))} \leq p(u^n|s^n) \leq 2^{-n(H(U|S)-\delta(\epsilon'))}
\end{align}
for some $\delta(\epsilon')\to 0$ as $\epsilon'\to 0$. This second assertion, i.e., \eqref{eq:f3}, follows from \cite[Lemma 12.3]{ElGamal:Network11}. Now, from \eqref{eq:f1}, \eqref{eq:f2}, \eqref{eq:f3}, and the fact that $V\to S \to U$ forms a Markov chain, we obtain from Markov lemma~\cite{ElGamal:Network11} that, for some $\epsilon''>\epsilon'$, $\lim\limits_{n\to\infty} \Pr\{V^n(L),S^n,U^n(M,K)\in\Tc_{\epsilon''}^{(n)}|V^n(L)=v^n,S^n=s^n\}=1$. Denoting $\tilde{\Ec}_6=\{V^n(L),S^n,U^n(M,K)\notin\Tc_{\epsilon''}^{(n)}\}$, the analysis above implies that $\Pr\{\Ec_1^c \cap \Ec_5^c \cap \tilde{\Ec}_6\}\to 0$ as $n\to\infty$. (We note that, an analysis similar to the one above is given for the Berger-Tung inner bound in~\cite{ElGamal:Network11}. The difference here is that the bin index of the covering codeword $V^n$ determines the $M$ index of $U^n$. Still, given a realization of $v^n$, and hence $m$, we have covering sequences, i.e., the set $\{U^n(m,k)\}|_{k=1}^{2^{n(R_u-R_r)}}$, for $s^n$, from which the argument follows.) Then, by conditional typicality lemma~\cite{ElGamal:Network11}, for some $\epsilon>\epsilon''$, we have $(V^n(L),S^n,U^n(M,K),Y^n) \in \Tc_{\epsilon}^{(n)}$, as $(V^n(L),S^n,U^n(M,K))\in\Tc_{\epsilon''}^{(n)}$ (due to having $\Ec_1^c\cap\Ec_5^c\cap\tilde{\Ec}_6$ w.h.p.), and $Y^n|(V^n(L)=v^n,S^n=s^n,U^n(M,K)=u^n)\sim \prod\limits_{i=1}^n p_{Y|U,S}(y_i|u_i,s_i)$. This implies that $\Pr\{\Ec_1^c\cap\Ec_5^c\cap\Ec_6\}\to 0$ as $n\to\infty$.

g) Finally, $\Pr\{\Ec_7\}$ analysis follows as that for the proof of Wyner-Ziv coding (see, e.g., \cite[Section 11.3.1]{ElGamal:Network11}). In particular, considering the event $\tilde{\Ec_7}=\{(V^n(l),U^n(M,K),Y^n)\in \Tc_\epsilon^{(n)} \textrm{ for some } l\in\Bc(1), M\neq1\}$, \cite[Lemma 11.1]{ElGamal:Network11} shows that $\Pr\{\Ec_7\}\leq \Pr\{\tilde{\Ec}_7\}$. Then, as each $V^n(l)$ with $l\in\Bc(1)$ is generated by $\prod\limits_{i=1}^n p_V(v_i)$ and independent of $(U^n(M,K),Y^n)$, from packing lemma~\cite{ElGamal:Network11}, $\Pr\{\tilde{\Ec}_7\}\to 0$ as $n\to \infty$, if
\begin{equation}\label{eq:RvminusRr}
R_v-R_r < I(V;U,Y)-\delta(\epsilon).
\end{equation}
Therefore, $\Pr\{\Ec_7\}\to 0$ as $n\to\infty$, if \eqref{eq:RvminusRr} holds.

The analysis above implies that $\Pr\{\Ec\}\to 0$ as $n\to\infty$ if \eqref{eq:RuminusRm}, \eqref{eq:Ru}, \eqref{eq:Rv}, and \eqref{eq:RvminusRr} hold. Here, we set $R_m=R_r$ and $R_u-R_r=I(U;S)+\delta_1$ and $R_v=I(V;S)+\delta_2$, for some arbitrarily small $\delta_1$ and $\delta_2$. This will satisfy \eqref{eq:RuminusRm} and \eqref{eq:Rv}. Furthermore, we set $R_r=I(V;S)-I(V;U,Y)+\delta_2=I(V;S|U,Y)+\delta_2$, which is the rate required to describe $V^n$ to the decoder. This will satisfy \eqref{eq:RvminusRr}. Finally, for a given $p(u|s)$, we choose some $p(v|s)$ to support transmission of the refinement message $b=m$ with rate $R_r=R_m<I(U;Y)-I(U;S)-\delta_1$. This will satisfy \eqref{eq:Ru}, as $R_u=R_r+(I(U;S)+\delta_1)<I(U;Y)$. We note that, $R_r\leq H(S|U,Y)+\delta_2$ as $I(V;S|U,Y)\leq H(S|U,Y)$ in above. Therefore, for a given $p(u|s)$, the proposed coding scheme supports any state refinement rate $R_r$ satisfying
\begin{equation*}
R_r \leq \min\{I(U;Y)-I(U;S), H(S|U,Y)\}.
\end{equation*}

Utilizing the analysis above, we now detail the derivation of the amplification rate. We start by defining an event indicator random variable $E$. Consider setting $E=1$, if a decoding error ($\Ec$) occurs or the state sequence observed at the encoder ($S^n$) turns out to be non-typical. Set $E=0$, otherwise. We continue with the following set of relations. (We drop the indicies of codewords.)
\begin{align*} 
\frac{1}{n}I(S^n;Y^n)
&\stackrel{(a)}{\geq} \frac{1}{n}I(S^n;Y^n|E)-\frac{1}{n}\\
&=\frac{1}{n}I(S^n;Y^n|E=0)\Pr\{E=0\}+\frac{1}{n}I(S^n;Y^n|E=1)\Pr\{E=1\}-\frac{1}{n}\\
&\stackrel{(b)}{\geq} \frac{1}{n}I(S^n;Y^n|E=0)\Pr\{E=0\}-\frac{1}{n}\\
&\stackrel{(c)}{=} \frac{1}{n}I(S^n;Y^n,U^n,V^n|E=0)(1-\Pr\{E=1\})-\frac{1}{n}\\
&\stackrel{(d)}{\geq} \frac{1}{n}I(S^n;Y^n,U^n,V^n|E=0) - (H(S)+\epsilon_1)\Pr\{E=1\})-\frac{1}{n}\\
&\stackrel{(e)}{\geq} (H(S)-\epsilon_1) - \frac{1}{n}H(S^n|Y^n,U^n,V^n,E=0) - \epsilon_2\\
&\stackrel{(f)}{\geq} (H(S)-\epsilon_1) - H(S|Y,U,V) - \epsilon_2\\
&\stackrel{(g)}{=} I(S;Y,U)+R_r-(\epsilon_1 +\epsilon_2 + \delta_2),
\end{align*}
where (a) follows by the indicator event conditioning lemma given in Appendix~\ref{App:ConditioningLemma}, (b) follows as $I(S^n;Y^n|E=1)\geq 0$, (c) is due to the definition of $E$, wherein $E=0$ implies that $\Ec^c$, i.e., the decodability of the codewords $(U^n(M,K),V^n(L))$ form $Y^n$, (d) is by $I(S^n;Y^n,U^n,V^n|E=0)\leq H(S^n|Y^n,U^n,V^n,E=0)\leq H(S^n|E=0)\leq n(H(S)+\epsilon_1)$ for some arbitrarily small $\epsilon_1$ as $E=0$ implies that $S^n$ is typical (and $S^n$ is generated i.i.d. here), (e) is by taking $\epsilon_2=(H(S)+\epsilon_1)\Pr\{E=1\})-\frac{1}{n}$, which can be made arbitrarily small as $n\to \infty$ (as $\Pr\{E=1\}\to 0$ as $n\to\infty$), and lower bounding $H(S^n|E=0)\geq n(H(S)-\epsilon_1)$, which follows as $E=0$ implies that $S^n$ is typical (and $S^n$ is generated i.i.d. here), (f) follows as $H(S^n|Y^n,U^n,V^n,E=0)=\sum\limits_{i=1}^n H(S_i|S_1^{i-1},Y^n,U^n,V^n,E=0)\leq \sum\limits_{i=1}^n H(S_i|Y_i,U_i,V_i)=n(H(S|Y,U,V))$, and (g) is by $H(S)-H(S|Y,U,V)=I(S;Y,U,V)=I(S;Y,U)+I(S;V|Y,U)=I(S;Y,U)+R_r-\delta_2$. (This follows, as we set $R_r=I(S;V|Y,U)+\delta_2$ in the coding scheme.) From the last expression, we identify that any $R_a\leq I(S;Y,U)+R_r$ is achievable.

\textit{Leakage rate analysis:} We first show the achievability of $R_l\geq I(U,S;Z)$ with the following.
\begin{align*}
\frac{1}{n}I(S^n;Z^n)
&\leq \frac{1}{n}I(U^n,S^n;Z^n)\\
&= \frac{1}{n}H(Z^n)-\frac{1}{n}H(Z^n|U^n,S^n)\\
&= \frac{1}{n}\sum\limits_{i=1}^n H(Z_i|Z_1^{i-1}) - \frac{1}{n}\sum\limits_{i=1}^n H(Z_i|Z_1^{i-1},U^n,S^n) \\
&\stackrel{(a)}{\leq} \frac{1}{n}\sum\limits_{i=1}^n H(Z_i) - \frac{1}{n}\sum\limits_{i=1}^n H(Z_i|U_i,S_i) \\
&\stackrel{(b)}{=}  I(U,S;Z),
\end{align*}
where (a) follows as $H(Z_i|Z_1^{i-1})\leq H(Z_i)$ due to the fact that conditioning does not increase the entropy and $H(Z_i|Z_1^{i-1},U^n,S^n)=H(Z_i|U_i,S_i)$ as $(Z_1^{i-1},U_1^{i-1},S_1^{i-1},U_{i+1}^n,S_{i+1}^n) \to (U_i,S_i) \to Z_i$ forms a Markov chain (this is due to i.i.d. generation of $x_i$ from $(u_i,s_i)$, and i.i.d. generation of $z_i$ from $(x_i,s_i)$), and (b) follows as $\frac{1}{n}\sum\limits_{i=1}^n I(U_i,S_i;Z_i) = I(U,S;Z)$. (We note that single letterization argument for $H(Z^n|U^n,S^n)$ above is also given in (26) of \cite{Merhav:Information07}.)

Next, we focus on the achievability of $R_l\geq I(S;Z,U)+R_r$. We have the following.
\begin{align*}
\frac{1}{n}I(S^n;Z^n)
&\leq \frac{1}{n}I(U^n(M,K),S^n;Z^n)\\
&\stackrel{(a)}{=} \frac{1}{n}I(U^n(M,K),M;Z^n) + \frac{1}{n}I(S^n;Z^n|U^n(M,K)) \\
&\stackrel{(b)}{\leq} \frac{1}{n}H(M) + \frac{1}{n}H(U^n(M,K)|M) + \frac{1}{n}I(S^n;Z^n|U^n(M,K)) \\
&=\frac{1}{n}H(M) + \frac{1}{n}H(U^n(M,K)|M) + \frac{1}{n}H(Z^n|U^n(M,K)) - \frac{1}{n}H(Z^n|U^n(M,K),S^n)\\
&\stackrel{(c)}{\leq} R_r + I(U;S)+\delta_1 + H(Z|U) - H(Z|U,S) \\
&= I(S;Z,U) + R_r + \delta_1,
\end{align*}
where (a) is by adding $I(M;Z^n|U^n(M,K))=0$ as $M$ can be determined from $U^n(M,K)$ for a given codebook, (b) follows as $I(U^n(M,K),M;Z^n)\leq H(U^n(M,K),M) = H(M)+H(U^n(M,K)|M)$, (c) is due to having $H(M)\leq nR_r$ (a random variable with $2^{nR_r}$ values has entropy at most $nR_r$), $H(U^n(M,K)|M)\leq n(R_u-R_r)=n(I(U;S)+\delta_1)$ (here, given $M$, there are $2^{n(R_u-R_r)}$ number of $U^n$ codewords, and the entropy of this set is maximized if the index $K$ has the uniform distribution), $H(Z^n|U^n)=\sum\limits_{i=1}^n H(Z_i|Z_1^{i-1},U^n)\leq \sum\limits_{i=1}^n H(Z_i|U_i)=nH(Z|U)$, and having $H(Z^n|U^n,S^n)=\sum\limits_{i=1}^n H(Z_i|Z_1^{i-1},U^n,S^n)=\sum\limits_{i=1}^n H(Z_i|U_i,S_i)=nH(Z|U,S)$, where the second inequality holds as $(Z_1^{i-1},U_1^{i-1},S_1^{i-1},U_{i+1}^n,S_{i+1}^n) \to (U_i,S_i) \to Z_i$ forms a Markov chain (as detailed in the previous paragraph). (We note that single letterization arguments for $H(Z^n|U^n)$ and $H(Z^n|U^n,S^n)$ above are also given in (24) and (26) of \cite{Merhav:Information07}.)

This concludes the proof of Theorem~\ref{thm:Rin1}, where we take the union of the achievable pairs over all $p(u|s)\times p(x|u,s)=p(u,x|s)$ distributions.


\subsubsection{Message transmission}
We note that it is possible to allocate a part of the Gel'fand-Pinsker coding rate in the scheme proposed above in order to transmit messages. In particular, if $(R_a,R_l,R_r)$ satisfies the inequalities given in Theorem~\ref{thm:Rin1}, then $(R_a,R_l,R_m)$ is achievable with a message rate of $R_m=I(U;Y)-I(U;S)-R_r$. That is, the Gel'fand-Pinsker coding rate given by $I(U;Y)-I(U;S)$ can be divided into refinement rate $R_r$ and message rate $R_m$.


\subsubsection{State sequence covering}

In the region given by Theorem~\ref{thm:Rin1}, increasing  $R_r$ will not only increase the amplification rate but will also increase the leakage rate. Therefore, for some scenarios, implementing only a covering scheme might be advantageous. By choosing an arbitrarily small refinement rate $R_r$ in Theorem~\ref{thm:Rin1}, the following region can be achieved.

\begin{corollary}\label{thm:Rin2}[Covering]
Let $\Rc^{2}$ be the closure of the union of all $\left( R_a, R_l \right)$ pairs satisfying
\begin{eqnarray}
R_a &\leq& I(S;Y,U) \nonumber\\
R_l &\geq& \min \big\{I(U,S;Z), I(S;Z,U) \big\}\nonumber
\end{eqnarray}
over all distributions $p(u,x|s)$ satisfying $I(U;Y)> I(U;S)$.
Then, $\Rc^{2}\subseteq\Cc$.
\end{corollary}
We note that the rate of refinement, $R_r=I(V;S|U,Y)+\delta_2$, is set to an arbitrarily small value here, and the codeword $V^n$ only serves as a covering of the state sequence. Further implications of this covering scheme on the amplification-leakage region is discussed in Section~\ref{sec:RevDegDMC}. We note that, by transmission of a covering of the state, the leakage rate is shown to satisfy $R_l \geq \min \big\{I(U,S;Z), I(S;Z,U) \big\}$. Remarkably, one can guarantee such a bound, even if some state dependent information is transmitted over the channel. In particular, if the channel seen by Bob is stronger than the one seen by Eve, one can send the refinement information securely over the state dependent channel. This approach is detailed in the next section.


\subsection{Secure refinement}

Consider all input distributions $p(u,x|s)$ satisfying $I(U;Y)\geq I(U;Z)$. For such distributions, it is possible to send refinement information securely over the channel. This way, the leakage increase due to refinement index is decreased as the security of the index will lower the corresponding leakage rate achieved at Eve compared to the non-secured case. In the following, we first focus on transmission of secure messages over the state dependent channel, and then detail the proposed secure refinement approach.


\subsubsection{Secure message transmission over state dependent channels}
Consider that a transmitter wants to send a secure message $M$ over the state dependent channel in the presence of eavesdropper. This problem is studied in~\cite{Chen:Wiretap08}, and, we will revisit it here. In particular, we give a codebook construction and provide a lemma that upper bounds $I(M;Z^n)$, the leakage to the eavesdropper~\footnote{The codebook we provide here is a special case of the one proposed in~\cite{Chen:Wiretap08}, which considers an extended version for an equivocation rate analysis.}. This result is then utilized in the following part, in showing the proposed secure refinement approach.

\textit{Codebook Generation:} We divide the codebook construction in two parts, depending on whether $I(U;Z)> I(U;S)$ or not. If $I(U;Z)> I(U;S)$, generate $2^{nR_u}$ codewords $U^n(M,T,K)$, where $M\in[1:2^{nR_m}]$, $K\in[1:2^{n(I(U;S)+\delta)}]$, and $T\in[1:2^{n(R_u-R_m-I(U;S)-\delta)}]$. Here, $T$ is randomly selected. We set $R_m=I(U;Y)-I(U;Z)$ and $R_u=I(U;Y)-\delta$, which imply $H(T)=n(I(U;Z)-I(U;S)-2\delta)$. If $I(U;Z)\leq I(U;S)$, generate $2^{nR_u}$ codewords $U^n(M,K)$, where $M\in[1:2^{nR_m}]$, $K\in[1:2^{n(I(U;S)+\delta)}]$. We set $R_m=I(U;Y)-I(U;S)-2\delta$ and $R_u=I(U;Y)-\delta$. In both cases, $M$ is the secure message index, and $K$ is used as a covering index (similar to the previous section).

The above codebook construction is the same as that of Gel'fand-Pinsker codebook (described in the previous section), with the only difference being that, for probability distributions satisfying $I(U;Z)> I(U;S)$, we select a part of the message as random (represented as $T$ in the paragraph above). This enables to argue that the remaining part of the Gel'fand-Pinsker message (represented as $M$ in the paragraph above) is secure against the eavesdropper (in terms of the leakage rate). We have the following.

\begin{lemma}\label{thm:Leakage}
For the codebook generation given above, for some $\epsilon\to 0$ as $n\to \infty$,
\begin{align*}
I(M;Z^n)\leq nI(S;Z,U)+H(S^n|U^n,Z^n)-H(S^n|M,T)+n\epsilon,
\end{align*}
where $T$ is a random variable uniformly distributed on the set $[1:2^{n(I(U;Z)-I(U;S)-2\delta)}]$ for $I(U;Z)> I(U;S)$, and $T=\emptyset$ for $I(U;Z)\leq I(U;S)$.
\end{lemma}

\begin{IEEEproof}
See Appendix~\ref{App:Leakage}.
\end{IEEEproof}

This result is utilized in the next section for the secure refinement approach. Here, we note the following corollary, which is an alternate proof of security in state dependent wiretap channels.

\begin{corollary}\label{thm:WiretapState}
For the codebook generation given above, if $M$ is independent of $S^n$, then 
\begin{align*}
\frac{1}{n}I(M;Z^n)\leq \epsilon,
\end{align*}
for some $\epsilon\to 0$ as $n\to \infty$.
\end{corollary}

\begin{IEEEproof}
See Appendix~\ref{App:WiretapState}.
\end{IEEEproof}


\subsubsection{Secure refinement via secure and state dependent message transmission}

In the previous section, a refinement approach is proposed where a Gel'fand-Pinsker coded message is utilized to resolve some ambiguity regarding the channel state at Bob. In such an approach, the message rate is utilized as the bin index of a covering ($V^n$) codebook. (See proof of Theorem~\ref{thm:Rin1}.) Here, we consider transmission of this refinement message securely to Bob by utilizing the codebook construction given above. As the message rate is modified from $I(U;Y)-I(U;S)$ to secure message rate $I(U;Y)-\max\{I(U;S),I(U;Z)\}$, this modification results in a refinement rate of $R_r \leq \min\big\{I(U;Y)-\max\{I(U;S),I(U;Z)\},H(S|Y,U)\big\}$. However, the leakage rate due refinement is now independent of $R_r$. The corresponding region is given by the following.

\begin{theorem}\label{thm:Rin3}
Let $\Rc^3$ be the closure of the union of all $\left( R_a, R_l \right)$ pairs satisfying
\begin{eqnarray}
R_a &\leq& I(S;Y,U)+R_r\nonumber\\
R_l &\geq& \min \{I(U,S;Z), I(S;Z,U)\}\nonumber\\
R_r &\leq& \min\big\{I(U;Y)-\max\{I(U;S),I(U;Z)\},H(S|Y,U)\big\}\nonumber
\end{eqnarray}
over all distributions $p(u,x|s)$ satisfying $I(U;Y)\geq I(U;S)$
Then, $\Rc^3\subseteq\Cc$.
\end{theorem}

Note that, the amplification rate that can be obtained with such an approach is lower than the previous case (Theorem~\ref{thm:Rin1}), as the message rate $R_r\leq I(U;Y)-I(U;Z) < I(U;Y)-I(U;S)$ if $I(U;Z)> I(U;S)$. Therefore, the improvement on the leakage expression compared to Theorem~\ref{thm:Rin1} is obtained with a degradation on the amplification rate.

\begin{IEEEproof}
We use the codebook generation given above. The only difference compared to the codebook construction utilized in the proof of Theorem~\ref{thm:Rin1} is that part of the message (i.e., $T$) is selected as random when $I(U;Z)> I(U;S)$. Therefore, the amplification rate analysis is the same as that of the one given in the proof of Theorem~\ref{thm:Rin1} with $R_r$ bounded by $R_r\leq I(U;Y)-\max\{I(U;S),I(U;Z)\}$ instead of $R_r\leq I(U;Y)-I(U;S)$. The proof of $R_l \geq I(U,S;Z)$ follows from the same steps given in the proof of Theorem~\ref{thm:Rin1}. Here, we show that $R_l \geq I(S;Z,U)$ holds as well.

Consider the following.
\begin{align*}
I(S^n;Z^n) 
&\leq I(M,S^n;Z^n)\\
&= I(M;Z^n) + I(S^n;Z^n|M)\\
&\stackrel{(a)}{\leq} [nI(S;Z,U)+H(S^n|U^n,Z^n)-H(S^n|M,T)+n\epsilon] + H(S^n|M) - H(S^n|M,Z^n)\\
&= nI(S;Z,U)+n\epsilon + H(S^n|U^n,Z^n) - H(S^n|M,Z^n) + H(S^n|M) - H(S^n|M,T)\\
&\stackrel{(b)}{\leq} nI(S;Z,U)+n\epsilon + I(S^n;T|M)\\
&\stackrel{(c)}{=} nI(S;Z,U)+n\epsilon,
\end{align*}
where (a) follows from Lemma~\ref{thm:Leakage}, (b) follows as $H(S^n|U^n,Z^n)=H(S^n|U^n,M,Z^n)\leq H(S^n|M,Z^n)$ and $H(S^n|M) - H(S^n|M,T)=I(S^n;T|M)$, and (c) follows by observing $I(S^n;T|M)=0$, where $T=\emptyset$ for $I(U;Z)\leq I(U;S)$; and, $T$ is independently generated random variable for $I(U;Z)>I(U;S)$, implying that $H(T|M)=H(T|M,S^n)=H(T)$. As $\epsilon$ can be made arbitrarily small, we conclude from last expression that any $R_l\geq I(S;Z,U)$ is achievable.
\end{IEEEproof}


\section{Outer Bounds}
\label{sec:OuterBound}

In this section, we provide outer bounds on the achievable amplification-leakage rate region. In particular, we derive regions denoted by $\Rc_o$, to which any achievable $(R_a,R_l)$ must belong.

\begin{proposition}\label{thm:Rout1}
If $(R_a,R_l)$ is achievable, then $(R_a,R_l)\in \Rc_o^{1}$, where $\Rc_o^{1}$ is the closure of the union of all $(R_a,R_l)$ pairs satisfying
\begin{eqnarray}
R_a &\leq& \min \left\{H(S), I(X,S;Y) \right\}\nonumber\\
R_l &\geq& I(S;Z,U)\nonumber
\end{eqnarray}
over all $p(u,x|s)$ distributions satisfying $I(U;Z)\geq I(U;S)$.
\end{proposition}

\begin{IEEEproof}
See Appendix~\ref{App:Rout1}.
\end{IEEEproof}

If the channel is degraded, wherein $p(y,z|x,s)=p(y|x,s)p(z|y)$, the following outer bound can be obtained.

\begin{proposition}\label{thm:Rout2}
If the channel satisfies $p(y,z|x,s)=p(y|x,s)p(z|y)$ and if $(R_a,R_l)$ is achievable, then $(R_a,R_l)\in \Rc_o^{2}$, where $\Rc_o^{2}$ is the closure of the union of all $(R_a,R_l)$ pairs sarisfying
\begin{eqnarray}
R_a &\leq& \min \left\{H(S), I(X,S;Y) \right\}\nonumber\\
R_l &\geq& I(S;Z,U)\nonumber\\
R_a-R_l &\leq& I(X,S;Y|Z)\nonumber
\end{eqnarray}
over all $p(u,x|s)$ distributions satisfying $I(U;Y)\geq I(U;S)$.
\end{proposition}

\begin{IEEEproof}
See Appendix~\ref{App:Rout2}.
\end{IEEEproof}
These outer bound regions are used in the following to establish special case results.


\section{Special discrete memoryless channel models}
\label{sec:DMC}


\subsection{Reversely degraded DMCs}
\label{sec:RevDegDMC}

We say that the channel is reversely degraded if $(X,S)\to Z \to Y$ forms a Markov Chain. Note that, this corresponds to a stronger channel seen by Eve compared to that of Bob. Therefore, reversely degraded scenarios imply $C_d\leq 0$, meaning that the state knowledge at Bob is not higher than that of at Eve. We have the following result for this set of channels.

\begin{theorem}\label{thm:RevDegDMC}
The optimal differential amplification rate for reversely degraded DMCs is given by
\begin{eqnarray}
C_d &=& \max\limits_{p(x|s)} I(S;Y)-I(S;Z)\nonumber
\end{eqnarray}
\end{theorem}

\begin{IEEEproof}Achievability of the stated difference follows from Theorem~\ref{thm:Rin1} by substituting $U=\emptyset$. We provide the converse in Appendix~\ref{App:RevDegDMC}.
\end{IEEEproof}

Note that coding can not improve this difference as the channel is reversely degraded. Thus, coding might help to increase $R_a$ at the expense of possibly decreasing $R_a-R_l$ for the reversely degraded scenario. This is also the case for the covering scheme given by Corollary~\ref{thm:Rin2}. That is, $R_a$ vs. $R_a-R_l$ can be traded-off using different input distributions. ($U=\emptyset$ case will correspond to the maximum $R_a-R_l$, and achieve $C_d$.)


\subsection{Modulo additive binary model}

Consider the channels given by
\begin{eqnarray*}
Y_i&=&X_i \oplus S_i \oplus N_i\nonumber\\
Z_i&=&X_i \oplus S_i \oplus \tilde{N}_i,
\end{eqnarray*}
where the state and noise distributions are generated i.i.d. as $S_i\sim \textrm{Bern}(p_s)$, $N_i\sim\textrm{Bern}(p_n)$, $\tilde{N}_i\sim\textrm{Bern}(p_{n_z})$. (All $p_k$s satisfy $p_k\in[0,0.5]$ for $k\in\{s,n,n_z\}$.) In this section, we use the following notation for the binary convolution $p\otimes q \triangleq p(1-q) + q(1-p)$.


\subsubsection{State cancellation scheme}

To cancel the state from the channel, we send
\begin{eqnarray*}
X_i=U_i\oplus S_i,
\end{eqnarray*}
where $U_i\sim \textrm{Bern}(p_u)$ and the codewords $U^n$ carry a description of the state sequence $S^n$. This way, we achieve the following inner-bound.

\begin{corollary}\label{thm:InnerStateCancellationBinary}
The state cancellation scheme, which sends $\textrm{Bern}(p_u)$ distributed signal XORed with state sequence at each time instant, achieves the set of $\left( R_a, R_l\right)$ pairs denoted by the region $\Rc^{\textrm{SC}}$, where
\begin{eqnarray*}
\Rc^{\textrm{SC}} = \textrm{Convex Hull} \left\{
\bigcup\limits_{p_u\in[0,0.5],p_u\otimes p_s \leq 0.5}
\left( R_a(p_u),
R_l(p_u)\right)
\right\} \subseteq \Cc,
\end{eqnarray*}
with
\begin{eqnarray*}
R_a(p_u) &\leq& \min \left\{ H(p_s), H(p_u \otimes p_n) - H(p_n) \right\} \\
R_l(p_u) &\geq& H(p_u \otimes p_{n_z}) - H(p_{n_z}).
\end{eqnarray*}
\end{corollary}

\begin{IEEEproof}
Achievability follows from Corollary~\ref{thm:Rin2}.
\end{IEEEproof}


\subsubsection{Optimal differential amplification rate}

\begin{corollary}
If $p_{n}\leq p_{n_z}$ and $H(p_s) \geq 1-H(p_n)$ for a binary model, the optimal amplification and leakage rate difference is given by 
\begin{equation}
C_d = H(p_{n_z}) - H(p_n). \nonumber
\end{equation}
\end{corollary}

\begin{IEEEproof}
From Proposition~\ref{thm:Rout2}, we obtain the following. If $p_{n}\leq p_{n_z}$, any given $(R_a,R_l)\in \Cc$ satisfies
\begin{equation*}
R_a - R_l \leq H(p_{n_z}) - H(p_n) +
\max\limits_{p(x|s)}
\left\{ H(X \oplus S \oplus N) - H(X \oplus S \oplus N_z)
\right\}.
\end{equation*}
Note that, this upper-bound can be evaluated by observing

$\max\limits_{p(x|s)}
\left\{ H(X \oplus S \oplus N) - H(X \oplus S \oplus N_z)
\right\} $
\begin{eqnarray*}
&=&
\max\limits_{p(x|s)}
\left\{ H(X \oplus S \oplus N) - H(X \oplus S \oplus N \oplus N_z^*)
\right\}\nonumber\\
&\leq&
\max\limits_{p(x|s)}
\left\{ H(X \oplus S \oplus N) - H(X \oplus S \oplus N \oplus N_z^* | N_z^*)
\right\}\nonumber\\
&=& 0,
\end{eqnarray*}
where the equality is due to the channel degradedness condition with appropriate noise term $N_z^*$ independent of $N$ such that $N \oplus N_z^*=N_z$, and the inequality is due to the fact that conditioning does not increase the entropy. Using this we observe that the outer-bound is maximized with a choice of $p(x)=0.5$, which evaluates to 
\begin{equation*}
R_a - R_l \leq H(p_{n_z}) - H(p_n).
\end{equation*}
This expression is achieved by Corollary~\ref{thm:InnerStateCancellationBinary}, when we choose
$p(u)=0.5$, if $H(p_s)\geq 1-H(p_n)$.
\end{IEEEproof}


\subsection{Memory with defective cells model}

We consider the model of information transmission over write-once memory device with stuck-at defective cells~\cite{Kuznetsov:Coding74,Heegard:Capacity81,Heegard:capacity83}. In this channel model, each memory cell corresponds to a channel state instant with cardinality $|\Sc|=3$, where the binary channel output is determined from the binary channel input and the channel state as:

\begin{figure}[t]
    \centering
    \includegraphics[width=0.5\columnwidth]{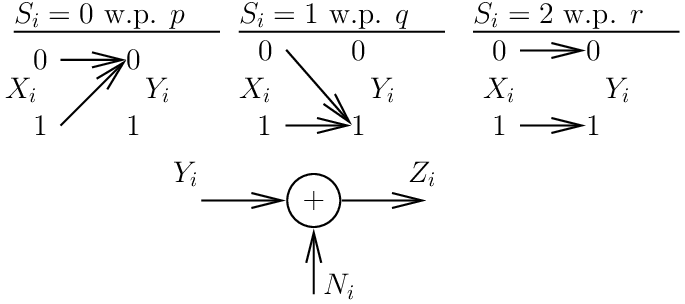}
    \caption{Channel model of memory with defective
    cells. $p=\Prob\{S=0\}$ (probability of being stuck at $0$),
    $q=\Prob\{S=1\}$ (probability of being stuck at $1$),
    $r=\Prob\{S=2\}$ (probability of being in a noiseless state),
    and $N\sim Bern(n)$, where $n\in[0,0.5]$ is the cross over
    probability of the BSC from $Y$ to $Z$.
    }
    \label{fig:ModelMemDef}
\end{figure}

\begin{eqnarray*}
p(y=0|x,s=0) &=& 1\nonumber\\
p(y=1|x,s=1) &=& 1\\
p(y=x|x,s=2) &=& 1,\nonumber
\end{eqnarray*}
where $\Prob\{S=0\}=p$ is the probability that the channel is stuck at $0$, $\Prob\{S=1\}=q$ is the probability that the channel is stuck at $1$, and $\Prob\{S=2\}=r$ is the probability of having a good channel where $y=x$ with $p+q+r=1$. We consider a binary symmetric channel (BSC) from $Y$ to $Z$, where
\begin{eqnarray*}
Z_i=Y_i \oplus N_i,
\end{eqnarray*}
with $N_i\sim Bern(n)$ for some $n\in[0,0.5]$. This corresponds to a degraded DMC model. (See Fig.~\ref{fig:ModelMemDef}.)

We present numerical results for this channel model with three regions: Uncoded region, coded region, and an outer-bound region. The uncoded region is obtained by setting $U=\emptyset$ in Theorem~\ref{thm:Rin1}, where we have the set of $(R_a,R_l)$ pairs satisfying
\begin{eqnarray*}
R_a &\leq& I(S;Y)\\
R_l &\geq& I(S;Z)
\end{eqnarray*}
over all possible $p(x|s)$. For the coded region, we simulate a sub-region of the one given in Theorem~\ref{thm:Rin1}, where we set $U=Y$ and achieve the set of $(R_a,R_l)$ pairs satisfying
\begin{eqnarray*}
R_a &\leq& \min\{H(S),H(Y)\}\\
R_l &\geq& I(Y,S;Z)=H(Z)-H(N)
\end{eqnarray*}
over all possible $p(x|s)$. For converse arguments, we consider the outer-bound region given by the set of $(R_a,R_l)$ pairs satisfying
\begin{eqnarray*}
R_a &\leq& \min\{H(S),I(X,S;Y)=H(Y)\}\\
R_l &\geq& I(S;Z)
\end{eqnarray*}
over all possible $p(x|s)$. This outer-bound region follows from Proposition~\ref{thm:Rout1}. We evaluate the regions above in terms of the channel parameters as follows. Let $\Prob\{X=1\}=\alpha$. Then,
\begin{eqnarray*}
H(S) &=& H(p,q,r), \\
H(Y|S) &=& r H(\alpha) \\
H(Y) &=& H(q+r\alpha) \\
H(Z|S) &=& (p+q)H(n)+r H(\alpha \otimes n) \\
H(Z) &=& H((q+r\alpha) \otimes  n),
\end{eqnarray*}
where $H(\cdot,\cdot,\cdot)$ is the ternary entropy function, $H(\cdot)$ is the binary entropy function, and $\otimes$ is the binary convolution given by $p\otimes q = p(1-q)+q(1-p)$. The numerical results are given in Fig.~\ref{fig:SimMemDef}. The regions are truncated with $R_l\leq H(S)$ as any $R_l>H(S)$ is trivially achievable. We note that, the coded region is potentially larger than its uncoded counterparts even when we only compute a subset of the coded achievable region. This shows the gains that can be leveraged by the proposed scheme, i.e., sending a refinement of the state sequence over the channel.

\begin{figure}[t]
    \centering
    \includegraphics[width=1\columnwidth]{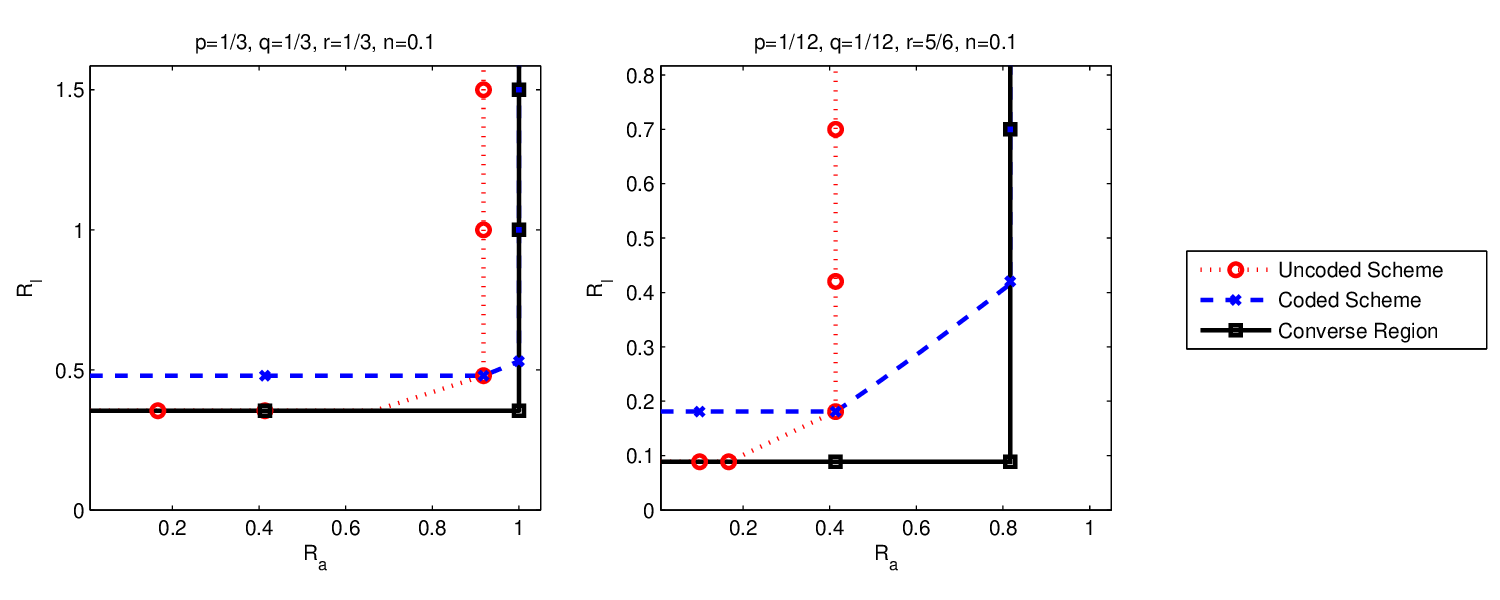}
    \caption{Simulation results for memory with defective
    cells model.
    }
    \label{fig:SimMemDef}
\end{figure}


\section{Gaussian Scenario}
\label{sec:Gauss}

Consider the channels given by
\begin{eqnarray*}
Y_i&=&X_i+S_i+N_i\nonumber\\
Z_i&=&X_i+S_i+\tilde{N}_i,
\end{eqnarray*}
where the state and noise distributions are generated i.i.d. as $S_i\sim\Nc(0,\sigma_s^2)$, $N_i\sim\Nc(0,\sigma_n^2)$, $\tilde{N}_i\sim\Nc(0,\sigma_{n_z}^2)$, and the cost constraint on the channel input is given by $c(x)=x^2$ and $C=P$, i.e., $\frac{1}{n}\sum\limits_{i=1}^n \EE \{X_i^2\}\leq P.$ (See Fig.~\ref{fig:ModelGauss}.)

\begin{figure}[t]
    \centering
    \includegraphics[width=0.3\columnwidth]{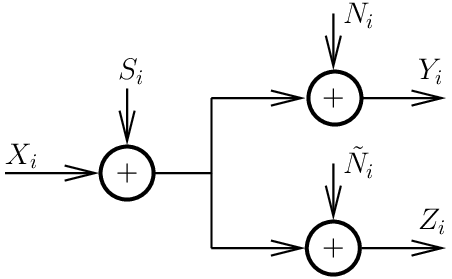}
    \caption{The channel model for the Gaussian setting.
    $S_i\sim \Nc(0,\sigma_s^2)$, $N_i\sim \Nc(0,\sigma_n^2)$,
    and
    $\tilde{N}_i\sim \Nc(0,\sigma_{n_z}^2)$.
    }
    \label{fig:ModelGauss}
\end{figure}


\subsection{An inner bound for $\Cc$ using an uncoded scheme}\label{sec:uncoded}

The inner bound is based on sending an amplified version of $S$ together with some additional Gaussian noise. This \emph{uncoded} signal is constructed as follows.
\begin{eqnarray}\label{eqn:uncoded}
X_i=\rho \frac{\sigma_x}{\sigma_s}S_i + \sqrt{(1-\rho^2)} \sigma_x T_i,
\end{eqnarray}
where $T_i\sim\Nc(0,1)$ independent of $S_i$, $\rho\in[-1,1]$, and $\sigma_x^2\leq P$. Here, $\rho^2$ is the fraction of the power allocated to $S_i$. This scheme achieves the following region.

\begin{theorem}\label{thm:UncodedGaussian}
The uncoded scheme, which forwards $S_i$ at each time step together with some i.i.d. Gaussian noise as given in \eqref{eqn:uncoded}, achieves the set of $\left( R_a, R_l\right)$ pairs denoted by the region $\Rc^{\textrm{uncoded}}$, where
\begin{eqnarray*}
\Rc^{\textrm{uncoded}} = \textrm{Convex Hull} \left\{
\bigcup\limits_{\rho\in[-1,1],\sigma_x^2\in[0,P]}
\left( R_a(\rho,\sigma_x), R_l(\rho,\sigma_x)\right)
\right\} \subset \Cc,
\end{eqnarray*}
with
\begin{eqnarray}
R_a(\rho,\sigma_x) &=& \frac{1}{2} \log\left(
1 + \frac{\sigma_s^2 + 2 \rho \sigma_s \sigma_x + \rho^2 \sigma_x^2}
{\sigma_n^2 + (1-\rho^2) \sigma_x^2}
\right)\label{eqn:rauncoded}\\
R_l(\rho,\sigma_x) &=& \frac{1}{2} \log\left(
1 + \frac{\sigma_s^2 + 2 \rho \sigma_s \sigma_x + \rho^2 \sigma_x^2}
{\sigma_{n_z}^2 + (1-\rho^2) \sigma_x^2}
\right). \label{eqn:rluncoded}
\end{eqnarray}
\end{theorem}
The expressions above are obtained by evaluating $R_a=I(S;Y)$ and $R_l=I(S;Z)$ on account of uncoded transmission in (\ref{eqn:uncoded}).

\textit{Examples:}
\begin{itemize}
\item If $P\geq \sigma_s^2$, one can set $X=-S$ and achieve the pair
\begin{equation*}
(R_a=0,R_l=0).
\end{equation*}
\item Another trivial point is obtained by setting $X=0$, which achieves
\begin{equation*}
\left(R_a=\frac{1}{2}\log\left(1+\frac{\sigma_s^2}{\sigma_n^2}\right),
R_l=\frac{1}{2}\log\left(1+\frac{\sigma_s^2}{\sigma_{n_z}^2}\right)\right).
\end{equation*}
\end{itemize}


\subsection{Outer-bounds on $\Cc$}

\begin{corollary}\label{thm:rarlbound}
Let $\rho$ denote the correlation coefficient between $X$ and $S$. The set of all achievable rate pairs $(R_a,R_l)$, satisfy

\begin{eqnarray*}
R_a &\leq& \frac{1}{2}\log\left(1+\frac{\sigma_s^2+\sigma_x^2+2\rho\sigma_s\sigma_x}{\sigma_n^2}\right)\label{eqn:rabound}\\
R_l &\geq& \frac{1}{2}\log\left(1+\frac{\sigma_s^2+\rho^2\sigma_x^2+2\rho\sigma_s\sigma_x}{\sigma_{n_z}^2+\sigma_x^2(1-\rho^2)}\right)
\end{eqnarray*}
for $-1\leq\rho\leq 1$ and $\sigma_x^2\leq P$.
\end{corollary}
\begin{IEEEproof}
Using Proposition~\ref{thm:Rout2}, we have
\begin{equation}
R_a \leq I(X,S;Y) = h(Y) - h(Y|X,S) = h(Y) - h(N) \leq \frac{1}{2} \log \left( 1 + \frac{\sigma_s^2 + 2 \rho \sigma_s \sigma_x + \sigma_x^2}{\sigma_n^2}\right).
\end{equation}
Using Proposition~\ref{thm:Rout2}, the linear estimate $\hat{S}(Z)=\frac{\mathbb{E}[SZ]}{\mathbb{E}[Z^2]}Z$, and the fact the conditioning does not increase entropy, we get
\begin{equation*}
R_l \geq I(S;Z,U)\geq I(S;Z) = h(S) - h(S|Z) = h(S) - h(S-\hat{S}(Z)|Z) \geq h(S) - h(S-\hat{S}(Z)).
\end{equation*}
Since the entropy maximizing distribution for a given second moment is a Gaussian, we have 
\begin{equation*}
h(S-\hat{S}(Z)) \leq \frac{1}{2} \log 2\pi e\left(\frac{\sigma_s^2}{1 + \frac{\sigma_s^2 + 2 \rho \sigma_s \sigma_x + \rho^2\sigma_x^2}{\sigma_{n_z}^2+\sigma_x^2(1-\rho^2)}}\right),
\end{equation*}
leading to 
\begin{equation*}
R_l\geq \frac{1}{2} \log \left( 1 + \frac{\sigma_s^2 + 2 \rho \sigma_s \sigma_x + \rho^2\sigma_x^2}{\sigma_{n_z}^2+\sigma_x^2(1-\rho^2)}\right).
\end{equation*}
\end{IEEEproof}

\begin{corollary}
Let $\rho$ denote the correlation coefficient between $X$ and $S$. If $\sigma_n^2\leq \sigma_{n_z}^2$, then the set of all achievable rate pairs $(R_a,R_l)$ satisfy
\begin{equation}
R_a - R_l \leq
\frac{1}{2}\log\left(
1+ \frac{\sigma_s^2 + 2 \rho \sigma_s \sigma_x + \sigma_x^2}{\sigma_n^2}
\right)
-
\frac{1}{2}\log\left(
1+ \frac{\sigma_s^2 + 2  \rho \sigma_s \sigma_x + \sigma_x^2}{\sigma_{n_z}^2}
\right)\label{eqn:diff},
\end{equation}
for $-1\leq\rho\leq 1$ and $\sigma_x^2\leq P$.
\end{corollary}
\begin{IEEEproof}
By Proposition~\ref{thm:Rout2}, we have
\begin{equation*}
R_a-R_l \leq I(X,S;Y|Z).
\end{equation*}
Without loss of generality, we consider $\tilde{N}=N+N'$ with $\sigma_{n_z}^2=\sigma_{n}^2+\sigma_{n'}^2$ where $N'$ is independent of $N$. Noting that,
\begin{equation*}
I(X,S;Y|Z) = h(Y|Z) - h(Y|X,S,Z) = h(Y|Z) - h(N|\tilde{N}),
\end{equation*}
we upper bound $h(Y|Z)$ using the following. Consider two zero-mean correlated random variables $A$ and $B$.
\begin{eqnarray*}
h(A|B) &\overset{(a)}{=}& h(A-\hat{A}(B)|B)\nonumber\\
&\leq& h(A-\hat{A}(B))\\
&\overset{(b)}{\leq}& \frac{1}{2}\log(2\pi e \sigma_e^2),\nonumber
\end{eqnarray*}
where in (a) we used $\hat{A}(B)$ as the estimate of $A$ given $B$, and (b) follows by defining the estimation error variance $\sigma_e^2\triangleq E\left[(A-\hat{A}(B))^2\right]$ and the fact that Gaussian distribution maximizes entropy given the variance. We then upper bound the optimal estimator error variance by the linear MMSE variance. Therefore,
\begin{equation*}
h(A|B) \leq \frac{1}{2}\log\left(2\pi e
\left(
\textrm{var}(A)-\frac{E\left[(AB)^2\right]}{\textrm{var}(B)}
\right)
\right).
\end{equation*}
Using the above, we obtain
\begin{eqnarray}
R_a-R_l &\leq& \frac{1}{2}\log\left(2\pi e
\left(
\sigma_s^2 + 2  \rho\sigma_s \sigma_x + \sigma_x^2 + \sigma_n^2
-\frac{(\sigma_s^2 + 2  \rho\sigma_s \sigma_x + \sigma_x^2 + \sigma_{n}^2)^2}
{\sigma_s^2 + 2  \rho\sigma_s \sigma_x + \sigma_x^2 + \sigma_{n}^2 + \sigma_{n'}^2}
\right)
\right)\nonumber\\
&&{}-
\frac{1}{2}\log\left(2\pi e
\left(
\sigma_n^2-\frac{(\sigma_n^2)^2}{\sigma_{n}^2 + \sigma_{n'}^2}
\right)
\right)\nonumber\\
&=&
\frac{1}{2}\log\left(
1+ \frac{\sigma_s^2 + 2  \rho\sigma_s \sigma_x + \sigma_x^2}{\sigma_n^2}
\right)
-
\frac{1}{2}\log\left(
1+ \frac{\sigma_s^2 + 2  \rho\sigma_s \sigma_x + \sigma_x^2}{\sigma_n^2+\sigma_{n'}^2}
\right).
\end{eqnarray}
This completes the proof.
\end{IEEEproof}


\subsection{Comparison of inner and outer bounds for the degraded Gaussian channel}
We now compare the uncoded scheme and the outer bound presented above. In particular, we show that the uncoded transmission scheme achieves certain corner points of the amplification-masking region and that the gap between the inner and outer bounds on the region is within $1/2$ bits for a wide set of channel parameters. We also show that the uncoded scheme achieves the optimal difference $R_a-R_l$.


\subsubsection{Characterization of the gap between achievable and converse regions}

We show that given any point $(R_a,R_l)$ in the converse region corresponding to a given $(\rho,\sigma_x)$, uncoded transmission achieves within $1/2$ bits of the converse region under certain conditions on channel parameters. In particular, for any given $R_a$, uncoded transmission achieves that $R_a$ and within $1/2$ bits of the bound on $R_l$ if $\frac{P}{\sigma_{n_z}^2}\leq 1$. Similarly for any given $R_l$, uncoded transmission achieves the given $R_l$ and within $1/2$ bits of the bound on $R_a$ if $\frac{P}{\sigma_n^2}\leq 1$. We prove these as follows. Using Corollary~\ref{thm:rarlbound}, any point in the outer bound region is described as
\begin{align}
R_a &= \frac{1}{2}\log\left(1+\frac{\sigma_s^2+\sigma_x^2+2\rho\sigma_s\sigma_x}{\sigma_n^2}\right)\nonumber\\
R_l &= \frac{1}{2}\log\left(1+\frac{\sigma_s^2+\rho^2\sigma_x^2+2\rho\sigma_s\sigma_x}{\sigma_{n_z}^2+\sigma_x^2(1-\rho^2)}\right)
\end{align}
for $-1\leq\rho\leq 1$ and $\sigma_x^2\leq P$. Now let us show that uncoded transmission achieves any $R_l$ in the region above and the gap from $R_a$ as above is within $1/2$ bits. Let the uncoded scheme be designed such that $X_i = \frac{\sigma_x}{\sigma_s}\rho S_i+T_i$, where $T_i\sim\mathcal{N}(0,\sigma_x^2(1-\rho^2))$ and independent of $S_i$. Now, by \eqref{eqn:rluncoded} and \eqref{eqn:rauncoded} this input achieves a leakage,  $I(S;Z)=\frac{1}{2}\log\left(1+\frac{\sigma_s^2+\rho^2\sigma_x^2+2\rho\sigma_s\sigma_x}{\sigma_{n_z}^2+\sigma_x^2(1-\rho^2)}\right)$ and $R_a$ given by $I(S;Y)=\frac{1}{2}\log\left(1+\frac{\sigma_s^2+\rho^2\sigma_x^2+2\rho\sigma_s\sigma_x}{\sigma_{n}^2+\sigma_x^2(1-\rho^2)}\right)$, which implies that the gap is given by
\begin{equation}\label{eqn:ragap}
I(X,S;Y) - I(S;Y) = I(X;Y|S) = \frac{1}{2}\log\left(1+\frac{\sigma_x^2(1-\rho^2)}{\sigma_{n}^2}\right)\leq \frac{1}{2},
\end{equation}
for $\frac{P}{\sigma_n^2}\leq 1$. Now, in order to prove the other claim, that uncoded achieves any given $R_a$ and the gap with $R_l$ is within $1/2$ bits, we proceed as follows. Given $R_a=\frac{1}{2}\log\left(1+\frac{\sigma_s^2+\sigma_x^2+2\rho\sigma_s\sigma_x}{\sigma_n^2}\right)$, we achieve this by choosing an uncoded scheme such that $X_i = \frac{\sigma_{x'}}{\sigma_s}S_i$ if $\rho\geq 0$ or $X_i = -\frac{\sigma_{x'}}{\sigma_s}S_i$ if $\rho< 0$. We choose $\sigma_{x'}$ such that
\begin{eqnarray*}
\sigma_{x'}^2+2\frac{\rho}{|\rho|}\sigma_s\sigma_{x'} &=& \sigma_x^2+2\rho\sigma_s\sigma_x, \quad \textrm{ if } \rho\neq 0,\\
\sigma_{x'}^2+2\sigma_s\sigma_{x'} &=& \sigma_x^2, \quad \textrm{ if } \rho= 0,
\end{eqnarray*}
where $0\leq\sigma_{x'}\leq \sqrt{P}$. By the intermediate value theorem for continuous functions, it is clear that there exists $\sigma_{x'}\leq \sqrt{P}$ such that the conditions above are satisfied. Further, the uncoded scheme achieves the desired $R_a$. For $\rho\neq 0$, the scheme achieves an $R_l$ given by $\frac{1}{2}\log\left(1+\frac{\sigma_s^2+\sigma_{x'}^2+2\frac{\rho}{|\rho|}\sigma_s\sigma_{x'}}{\sigma_{n_z}^2}\right)$ leading to a gap

\begin{align}
&\frac{1}{2}\log \left(1+\frac{\sigma_s^2+\sigma_{x'}^2+2\frac{\rho}{|\rho|}\sigma_s\sigma_{x'}}{\sigma_{n_z}^2}\right) - \frac{1}{2}\log\left(1+\frac{\sigma_s^2+\rho^2\sigma_x^2+2\rho\sigma_s\sigma_x}{\sigma_{n_z}^2+\sigma_x^2(1-\rho^2)}\right)\nonumber\\
&\quad\:= \frac{1}{2}\log \left(1+\frac{\sigma_s^2+\sigma_x^2+2\rho\sigma_s\sigma_x}{\sigma_{n_z}^2}\right) -\frac{1}{2}\log\left(1+\frac{\sigma_s^2+\rho^2\sigma_x^2+2\rho\sigma_s\sigma_x}{\sigma_{n_z}^2+\sigma_x^2(1-\rho^2)}\right)\nonumber\\
&\quad\:= \frac{1}{2}\log \left(\frac{\sigma_{n_z}^2+\sigma_s^2+\sigma_x^2+2\rho\sigma_s\sigma_x}{\sigma_{n_z}^2}\right) -\frac{1}{2}\log\left(\frac{\sigma_{n_z}^2+\sigma_s^2+\sigma_x^2+2\rho\sigma_s\sigma_x}{\sigma_{n_z}^2+\sigma_x^2(1-\rho^2)}\right)\label{eqn:rlgap}\nonumber\\
&\quad\:=  \frac{1}{2}\log \left(\frac{\sigma_{n_z}^2+\sigma_x^2(1-\rho^2)}{\sigma_{n_z}^2}\right) = \frac{1}{2}\log \left(1+\frac{\sigma_x^2(1-\rho^2)}{\sigma_{n_z}^2}\right) \leq \frac{1}{2}\log 2 = \frac{1}{2},
\end{align}
when $\frac{P}{\sigma_{n_z}^2}\leq 1$. Following the same steps, $\rho=0$ case can be shown as well, implying the characterization of the trade-off region within $1/2$ bits for $\frac{P}{\sigma_{n_z}^2}\leq 1$ and $\frac{P}{\sigma_{n}^2}\leq 1$.


\subsubsection{Differential amplification capacity}
Note that the uncoded transmission achieves the maximum $R_a-R_l$. The upper bound on $R_a-R_l$ in \eqref{eqn:diff} is maximized for $\sigma_x^2=P$ and $\rho=1$. This maximum difference between $R_a$ and $R_l$ is achieved by uncoded transmission corresponding to $X=\frac{\sqrt{P}}{\sigma_s}S$ in Theorem~\ref{thm:UncodedGaussian}, and is given by
\begin{equation*}
C_d =
\frac{1}{2}\log\left(
1+ \frac{(\sigma_s+\sqrt{P})^2}{\sigma_n^2}
\right)
-
\frac{1}{2}\log\left(
1+ \frac{(\sigma_s+\sqrt{P})^2}{\sigma_{n_z}^2}
\right).
\end{equation*}


\subsubsection{Corner points of the trade-off region}
Consider the corner points of the amplification-masking region. Inspecting (\ref{eqn:rabound}), we observe that the point in the outer bound region corresponding to maximum amplification is given by $\rho=1$. Clearly, from (\ref{eqn:ragap}) and (\ref{eqn:rlgap}), we see that the gap is zero for $\rho=1$. Similarly, consider the point corresponding to minimum leakage $R_l$ in the \emph{weak} and \emph{moderate} interference regimes as in \cite{Merhav:Information07}. These points correspond to $\rho=-1$ and we have $I(X,S;Y) = I(S;Y)$ and $I(X,S;Z) = I(S;Z)$, leading to the gap being zero. This is also verified by setting $\rho=-1$ in (\ref{eqn:ragap}) and (\ref{eqn:rlgap}).


\subsubsection{Numerical results}

We compare the uncoded region with an outer-bound region (given in Proposition \ref{thm:rarlbound})  in Fig.~\ref{fig:SimGauss}. The first case corresponds to a degraded scenario, where the gap between the regions is fairly small as expected from the analysis given above. However, for the reversely degraded scenario with larger power constraint $P$ compared to the state power $\sigma_s^2$, the gap is larger. In Fig.~\ref{fig:SimGauss2}, we plot the differential amplification capacity for a degraded channel ($\sigma_n^2=1$, $\sigma_{n_z}^2=5$) for a range of power constraints $P$ and different values of $\sigma_s^2$. Note that the differential amplification capacity saturates in the high SNR regime, and the effect of encoder in increasing $C_d$ is decreasing as the power of the additive state increases.

\begin{figure}[t]
    \centering
    \includegraphics[width=1\columnwidth]{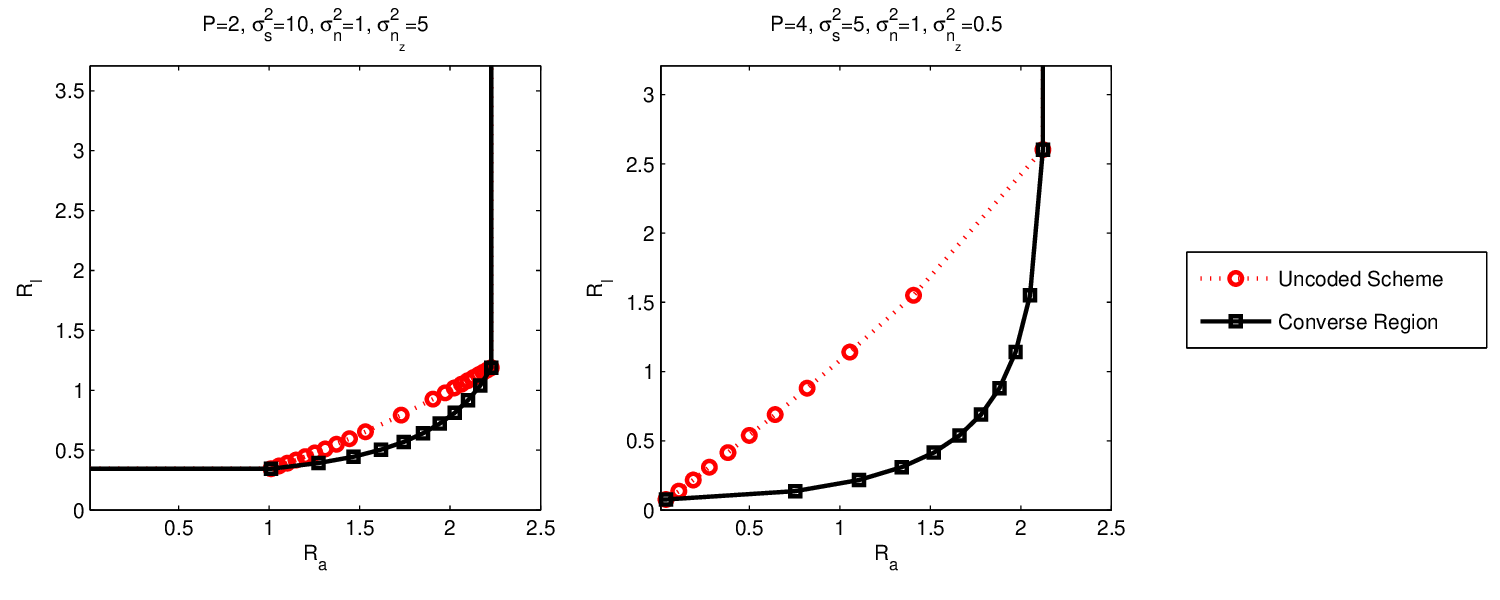}
    \caption{Simulation results for the Gaussian scenario.
    }
    \label{fig:SimGauss}
\end{figure}

\begin{figure}[t]
    \centering
    \includegraphics[width=1\columnwidth]{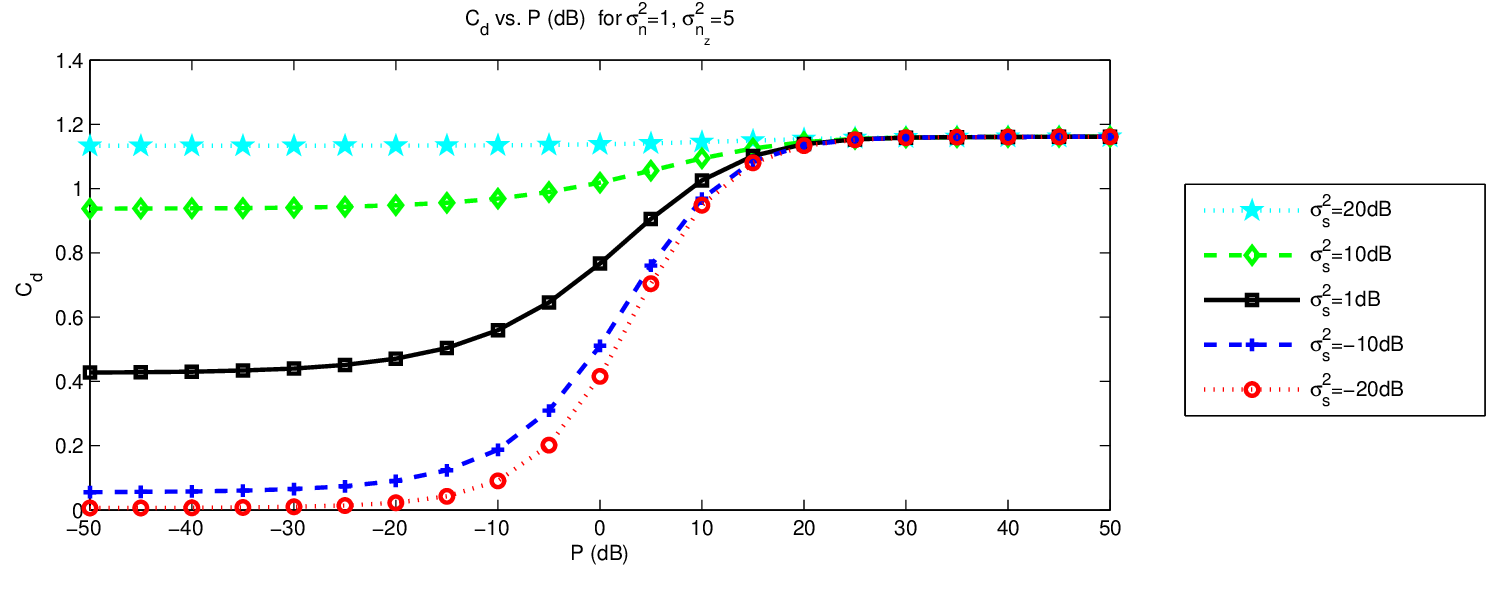}
    \caption{Differential amplification capacity $C_d$ vs. power
    $P$ (dB). The dashed curve with diamond markers correspond to
    the same scenario given in the degraded setting of Fig.~\ref{fig:SimGauss}
    ($\sigma_s^2=10$) for different power levels.
    }
    \label{fig:SimGauss2}
\end{figure}


\section{Conclusion}
\label{sec:Conclusion}

We study the problem of state amplification under the masking constraints, where the encoder (with the knowledge of non-causal state $S^n$) facilitates the amplification rate ($\frac{1}{n}I(S^n;Y^n)$) at Bob, which observes $Y^n$, while minimizing the leakage rate ($\frac{1}{n}I(S^n;Z^n)$) as much as possible at Eve, which observes $Z^n$. Our coding schemes are based on transmission of state dependent messages over the state dependent channel to Bob. The achievable region corresponding to this refinement strategy is derived by calculating bounds on amplification and masking rates. We also show that for the input distributions enabling Bob to be a ``stronger" receiver than Eve, the refinement information can be sent securely over the channel. This \emph{secure refinement} approach is shown to lead to non-trivial achievable regions. We also provided outer bounds, using which we showed that the scheme without secure refinement achieves the optimal $R_a-R_l$ over the region in the reversely degraded DMCs, the degraded binary channels, and Gaussian channels. For the degraded Gaussian model, we also characterized the optimal corner points, and the gap between the outer bound and  achievable regions.

Several interesting problems can be considered as future directions. First, the channel cost may be introduced for the DMC model as well, and the cost may have some dependence on the state sequence or vary according to a stochastic model. Second, causal channel state knowledge can be considered. Third, in addition to the task of state amplification and masking, transmission of messages to the receivers can be considered. Towards this end, signaling techniques for (secure) broadcast channel models can be utilized. Another extension direction is the coded state sequence setting~\cite{Heegard:capacity83}, which is a scenario that is more relevant to the broadcast and cognitive radio systems, where the coded signal (that carries a message from a codebook) corresponds to the channel state sequence that is non-causally known at the encoder. Finally, a source coding extension of the current model can be studied. In such a problem setup, the trade-off between the distortions achieved at both users can be analyzed. A relevant work for such an extension is \cite{Courtade:Information12}, where a source coding setup is considered with two sources (one has to be amplified and the other has to be masked) with a single receiver. We remark that the amplification and leakage rates analyzed in this paper can be utilized to provide a lower bound on the distortion that can be achieved at Bob and Eve, respectively, by evaluating the corresponding distortion-rate functions.


\section*{acknowledgement}
The authors are thankful to Yanling Chen of Ruhr-University Bochum, Yeow-Khiang Chia of Institute for Infocomm Research, Vincent Y. F. Tan of National University of Singapore, and anonymous reviewers for their valuable feedback.


\appendices


\section{$\Pr\{\Ec_3^c \cap \Ec_4\}\to 0$ as $n\to\infty$}
\label{App:CoveringIndexDecoding}

We use the arguments given in~\cite{Lim:Lossy10} in order to show this. Similar to the joint source-channel coding scenario studied in~\cite{Lim:Lossy10}, we have a codeword $U^n$ with a covering index $K$, which makes the codebook and covering index $K$ dependent to each other. In other words, $S^n$ sequence realization together with the codebook determines the index $K$, from which $(S^n,U^n)$ generates $Y^n$ through an i.i.d. generation process. Under such a scenario, as reported in~\cite{Lim:Lossy10}, decoding the index $K$ with $Y^n$ at receiver is successful if the number of $K$ indices is less than $2^{nI(U;Y)}$. We now provide this analysis here for completeness. Using the steps given in~\cite{Lim:Lossy10} for the scenario here, we have the following: For a given $M$, we have $2^{nR_k}$ number of $U^n(M,k)$ codewords ($k\in[1:2^{nR_k}]$), and the bound below.

\begin{align*}
\Pr\{\Ec_3^c \cap\Ec_4\} 
&=\Pr\{(U^n(M,k),Y^n)\in\Tc_\epsilon^{(n)} \textrm{ for some } k\neq K\}\\
&\stackrel{(a)}{\leq} \sum\limits_{k=1}^{2^{nR_k}} \Pr\{ (U^n(M,k),Y^n) \in\Tc_\epsilon^{(n)}, K\neq k \}\\
&= \sum\limits_{k=1}^{2^{nR_k}} \sum\limits_{s^n} p(s^n) \Pr\{ (U^n(M,k),Y^n) \in\Tc_\epsilon^{(n)}, K\neq k | S^n=s^n \}\\
&\stackrel{(b)}{=} 2^{nR_k} \sum\limits_{s^n} p(s^n) \Pr\{ (U^n(M,1),Y^n) \in\Tc_\epsilon^{(n)}, K\neq 1 | S^n=s^n \}\\
&\leq 2^{nR_k} \sum\limits_{s^n} p(s^n) \Pr\{ (U^n(M,1),Y^n) \in\Tc_\epsilon^{(n)}| K\neq 1 , S^n=s^n \}\\
&= 2^{nR_k} \sum\limits_{s^n} p(s^n) \sum\limits_{(u^n,y^n)\in\Tc_\epsilon^{(n)}}\Pr\{ U^n(M,1)=u^n,Y^n=y^n| K\neq 1 , S^n=s^n \}\\
&\stackrel{(c)}{=} 2^{nR_k} \sum\limits_{s^n} p(s^n) \sum\limits_{(u^n,y^n)\in\Tc_\epsilon^{(n)}} \sum\limits_{\bar{\Cc}} \Pr\{ U^n(M,1)=u^n,Y^n=y^n| K\neq 1 , S^n=s^n, \bar{\Cc}={\scriptstyle \bar{\Cc}}\}\\
& \quad \quad \times  \Pr\{\bar{\Cc}={\scriptstyle \bar{\Cc}} |K\neq 1 , S^n=s^n\}\\
&\stackrel{(d)}{=} 2^{nR_k} \sum\limits_{s^n} p(s^n) \sum\limits_{(u^n,y^n)\in\Tc_\epsilon^{(n)}} \sum\limits_{\bar{\Cc}} \Pr\{ U^n(M,1)=u^n | K\neq 1 , S^n=s^n, \bar{\Cc}={\scriptstyle \bar{\Cc}}\}\\
& \quad \quad \times  \Pr\{ Y^n=y^n | K\neq 1 , S^n=s^n, \bar{\Cc}={\scriptstyle \bar{\Cc}}\}\Pr\{\bar{\Cc}={\scriptstyle \bar{\Cc}} |K\neq 1 , S^n=s^n\}\\
&\stackrel{(e)}{\leq} 2^{nR_k} \sum\limits_{s^n} p(s^n) \sum\limits_{(u^n,y^n)\in\Tc_\epsilon^{(n)}} \sum\limits_{\bar{\Cc}} 2\Pr\{ U^n(M,1)=u^n \}\\
& \quad \quad \times  \Pr\{ Y^n=y^n | K\neq 1 , S^n=s^n, \bar{\Cc}={\scriptstyle \bar{\Cc}}\}\Pr\{\bar{\Cc}={\scriptstyle \bar{\Cc}} |K\neq 1 , S^n=s^n\}\\
&= 2^{nR_k} \sum\limits_{s^n} p(s^n) \sum\limits_{(u^n,y^n)\in\Tc_\epsilon^{(n)}} 2\Pr\{ U^n(M,1)=u^n \} \Pr\{ Y^n=y^n | K\neq 1 , S^n=s^n\}\\
&\stackrel{(f)}{\leq} 2^{nR_k} \sum\limits_{s^n} p(s^n) \sum\limits_{(u^n,y^n)\in\Tc_\epsilon^{(n)}} 4\Pr\{ U^n(M,1)=u^n \} \Pr\{ Y^n=y^n | S^n=s^n\}\\
&\stackrel{(g)}{=} 2^{nR_k+2} \sum\limits_{(u^n,y^n)\in\Tc_\epsilon^{(n)}} \prod\limits_{i=1}^n p_U(u_i) \Pr\{ Y^n=y^n\}\\
&\stackrel{(h)}{\leq} 2^{n(R_k-I(U;Y)+\delta)}\\
\end{align*}
where (a) is due to the union of events bound, and (b) follows by the symmetry of the codebook generation and coding, (c) is by defining $\bar{\Cc}=\{U^n(M,k), k\neq 1\}$, (d) is due to the fact that given $K\neq 1$, $U^n(M,1)\to (\bar{\Cc},S^n) \to Y^n$ forms a Markov chain, (e) is due to Lemma~\ref{thm:Lemma1MAC} given at the end of this section, (f) is due to Lemma~\ref{thm:Lemma2MAC} given at the end of this section, (g) is due to having i.i.d. generation for $U^n$, i.e., $\Pr\{ U^n(M,1)=u^n \}=\prod\limits_{i=1}^n p_U(u_i)$, (h) is due to joint typicality lemma~\cite{ElGamal:Network11}.

From the last expression, we obtain that $\Pr\{\Ec_3^c \cap \Ec_4\}\to 0$ as $n\to\infty$, if $R_k<I(U;Y)-\delta$. This implies the existence of sequence of codes implying the desired result as we set $R_k=I(U;S)+\delta<I(U;Y)-\delta$.

\begin{lemma}[Lemma 1 in~\cite{Lim:Lossy10}]\label{thm:Lemma1MAC}
For sufficiently large $n$, $\Pr\{ U^n(M,1)=u^n | K\neq 1 , S^n=s^n, \bar{\Cc}={\scriptstyle \bar{\Cc}}\}\leq 2\Pr\{ U^n(M,1)=u^n \}$
\end{lemma}

\begin{IEEEproof}
We have
\begin{align*}
\Pr\{ U^n(M,1)=u^n | K\neq 1 , S^n=s^n, \bar{\Cc}={\scriptstyle \bar{\Cc}}\}
&=\Pr\{ U^n(M,1)=u^n | S^n=s^n, \bar{\Cc}={\scriptstyle \bar{\Cc}}\}\\
&\quad \quad \times \frac{
\Pr\{ K\neq 1| U^n(M,1)=u^n  , S^n=s^n, \bar{\Cc}={\scriptstyle \bar{\Cc}}\}
}{
\Pr\{K\neq 1 | S^n=s^n, \bar{\Cc}={\scriptstyle \bar{\Cc}}\}
}\\
&\stackrel{(a)}{\leq}
\frac{
\Pr\{ U^n(M,1)=u^n \}
}{
\Pr\{K\neq 1 | S^n=s^n, \bar{\Cc}={\scriptstyle \bar{\Cc}}\}
}\\
&\stackrel{(b)}{\leq}
2\Pr\{ U^n(M,1)=u^n \}
\end{align*}
where (a) is due to independence of $U^n(M,1)$ and $(S^n,\bar{\Cc})$ together with bounding $\Pr\{ K\neq 1| U^n(M,1)=u^n  , S^n=s^n, \bar{\Cc}={\scriptstyle \bar{\Cc}}\}\leq 1$, and (b) follows from $\Pr\{K\neq 1 | S^n=s^n, \bar{\Cc}={\scriptstyle \bar{\Cc}}\}\geq \frac{1}{2}$, as shown below.

Consider $t=t({\scriptstyle \bar{\Cc}},s^n)=|\{u^n(M,k)\in{\scriptstyle \bar{\Cc}}: (u^n(M,k),s^n)\in\Tc_{\epsilon'}^{(n)}\}|$. Then, if $t\geq 1$, by the symmetry of the codebook generation and coding,
\begin{align*}
\Pr\{K= 1 | S^n=s^n, \bar{\Cc}={\scriptstyle \bar{\Cc}}\}
=\frac{\Pr\{(U^n(M,1),s^n)\in\Tc_{\epsilon'}^{(n)}\}}{t+1}
\leq \frac{1}{t+1}
\leq \frac{1}{2},
\end{align*}
where we upper bound the probability with $1$ and used $t\geq 1$. On the other hand, if $t=0$, for sufficiently large $n$, and due to the symmetry of the codebook generation and coding, we have
\begin{align*}
\Pr\{K= 1 | S^n=s^n, \bar{\Cc}={\scriptstyle \bar{\Cc}}\}
&\leq\Pr\{(U^n(M,1),s^n)\in\Tc_{\epsilon'}^{(n)}\} + \frac{\Pr\{(U^n(M,1),s^n)\notin\Tc_{\epsilon'}^{(n)}\}}{2^{nR_k}}\\
&\leq \Pr\{(U^n(M,1),s^n)\in\Tc_{\epsilon'}^{(n)}\} + \frac{1}{2^{nR_k}}\\
&\leq 2^{-n(I(U;S)-\delta(\epsilon'))} + \frac{1}{2^{nR_k}}\\
&\leq \frac{1}{2},
\end{align*}
where bound the probability with $1$ and utilized the joint typicality lemma~\cite{ElGamal:Network11}. The last inequality holds in the limit of large $n$ as $R_k>0$ and $I(U;S)>\delta(\epsilon')$, where $\delta(\epsilon')\to 0$ as $\epsilon'\to 0$.
\end{IEEEproof}

\begin{lemma}[Lemma 2 in~\cite{Lim:Lossy10}]\label{thm:Lemma2MAC}
For sufficiently large $n$, $\Pr\{ Y^n=y^n | K\neq 1 , S^n=s^n\}\leq 2 \Pr\{ Y^n=y^n | S^n=s^n\}$
\end{lemma}

\begin{IEEEproof}
We have
\begin{align*}
\Pr\{ Y^n=y^n | K\neq 1 , S^n=s^n\} 
&= \frac{\Pr\{ Y^n=y^n | S^n=s^n\} \Pr\{K\neq 1|S^n=s^n,Y^n=y^n\}}{\Pr\{K\neq 1|S^n=s^n\}}\\
&\stackrel{(a)}{\leq} \frac{\Pr\{ Y^n=y^n | S^n=s^n\} }{\Pr\{K\neq 1|S^n=s^n\}}\\
&\stackrel{(b)}{\leq} 2\Pr\{ Y^n=y^n | S^n=s^n\},
\end{align*}
where (a) follows as $\Pr\{K\neq 1|S^n=s^n,Y^n=y^n\}\leq 1$, and (b) is due to having $\Pr\{K\neq 1|S^n=s^n\}\geq 1/2$ for sufficiently large $n$ due to the symmetry of the codebook generation and coding.
\end{IEEEproof}


\section{Indicator event conditioning lemma}
\label{App:ConditioningLemma}

\begin{lemma}\label{thm:EventCond}
Consider an indicator random variable $E$, where $E=1$ for $\Ec$, and $E=0$ for $\Ec^c$, for an event $\Ec$. Then, for any $A,B$,
\begin{align*}
H(A|B) &\leq H(A|B,E)+1 \\
I(A;B) &\geq I(A;B|E)-1\\
I(A;B) &\leq I(A;B|E)+1 
\end{align*}
\end{lemma}
\begin{IEEEproof}
\begin{align}
H(A|B) &\stackrel{(a)}{\leq} H(A|B)+H(E|B,A)=H(E|B)+H(A|B,E)\stackrel{(b)}{\leq} H(A|B,E)+1 \label{eq:b1}\\
I(A;B) &\stackrel{(c)}{\geq} H(A|E)-H(A|B) \stackrel{(d)}{\geq} I(A;B|E)-1\nonumber\\
I(A;B) &\stackrel{(c)}{\leq} H(A)-H(A|B,E) \stackrel{(e)}{\leq} I(A;B|E)+1,\nonumber
\end{align}
where (a) is due to $H(E|B,E)\geq 0$, (b) is due to upper bound on the entropy of a binary random variable, (c) follows as conditioning does not increase entropy, (d) follows by \eqref{eq:b1}, and (e) follows by considering $B=\emptyset$ in \eqref{eq:b1} to upper bound $H(A)$.
\end{IEEEproof}


\section{Proof of Lemma~\ref{thm:Leakage}}
\label{App:Leakage}

\begin{IEEEproof}
We first consider $I(U;Z)>I(U;S)$ case, for which the codewords are represented by $U^n(M,T,K)$.
\begin{align*}
I(M;Z^n)
&=H(M)-H(M|Z^n)\\
&=H(M)-H(T,K,Z^n,M)+H(Z^n)+H(T,K|Z^n,M)\\
&=H(Z^n)+H(T,K|Z^n,M)-H(T,K|M)-H(Z^n|M,T,K)\\
&\stackrel{(a)}{\leq} n(H(Z)+\epsilon_1) - H(T|M) - H(U^n|M,T) - H(Z^n|M,T,U^n)\\
&\stackrel{(b)}{=} n(H(Z|U)+I(U;S)+\epsilon_2) - H(U^n,Z^n|M,T)\\
&= n(H(Z|U)+I(U;S)+\epsilon_2) - H(U^n,Z^n|M,T,S^n) - I(S^n;U^n,Z^n|M,T)\\
&\stackrel{(c)}{\leq} n(H(Z|U)+I(U;S)+\epsilon_2) - H(Z^n|M,T,S^n,U^n) - H(S^n|M,T)+H(S^n|U^n,Z^n)\\
&\stackrel{(d)}{\leq} n(I(S;Z,U)+\epsilon_2) - H(S^n|M,T)+H(S^n|U^n,Z^n),
\end{align*}
where (a) follows by $H(Z^n)=\sum\limits_{i=1}^n H(Z_i|Z_{1}^{i-1}) \leq \sum\limits_{i=1}^n H(Z_i) = nH(Z)$, $H(T,K|Z^n,M)\leq n\epsilon_1$ for some $\epsilon_1\to 0$ as $n\to\infty$ (this is decoding of $(T,K)$ at eavesdropper having $(Z^n,M)$, following from steps similar to the proof of Theorem~\ref{thm:Rin1} replacing Bob with Eve, as number of $(T,K)$ indices is $2^{n(I(U;Z)-\delta)}$, see also Appendix~\ref{App:CoveringIndexDecoding}), $H(K|M,T)=H(U^n|M,T)$, and having $H(Z^n|M,T,K)=H(Z^n|M,T,U^n)$ as $(M,T,K)$ is a one-to-one function of $(M,T,U^n(M,T,K))$, (b) follows by $H(T|M)=H(T)=n(I(U;Z)-I(U;S)-2\delta)$ as $T$ is independent of $M$ and uniformly random and taking $\epsilon_2=\epsilon_1+2\delta$, (c) follows as $H(U^n|M,T,S^n)\geq 0$ and $H(S^n|M,T,U^n,Z^n)=H(S^n|U^n,Z^n)$ as $(M,T)$ uniquely determined given $U^n(M,T,K)$, (d) is by having $$H(Z^n|M,T,S^n,U^n)=\sum\limits_{i=1}^n H(Z_i|Z_{1}^{i-1},M,T,S^n,U^n)=\sum\limits_{i=1}^n H(Z_i|S_i,U_i)=nH(Z|S,U),$$ which is due to the Markov chain $(Z_{1}^{i-1},M,T,S_1^{i-1},S_{i+1}^n,U_1^{i-1},U_{i+1}^n) \to (S_i,U_i)\to Z_i$ as $(U_i,S_i)$ generates $(X_i,S_i)$ which generates $Z_i$ i.i.d. due to the memoryless channel.

Secondly, consider $I(U;Z)\leq I(U;S)$ case, for which the codewords are represented by $U^n(M,K)$. We consider $K=[K_1,K_2]$ with $K_1=[1:2^{n(I(U;S)-I(U;Z)+2\delta)}]$ and $K_2=[1:2^{n(I(U;Z)-\delta)}]$, which together represent the covering index $K$. (This can be obtained via random binning of $2^{n(I(U;S)+\delta)}$ number of codewords $U^n(M,k)$, $k\in[1:2^{n(I(U;S)+\delta)}]$, into bins represented by $k_1$, where the codeword index per bin represented by $k_2$.) We continue as follows.
\begin{align*}
I(M;Z^n)
&=H(M)-H(M|Z^n)\\
&=H(M)-H(M,K_1,K_2,Z^n)+H(Z^n)+H(K_1,K_2|Z^n,M)\\
&=H(Z^n)+H(K_1|Z^n,M)+H(K_2|Z^n,M,K_1)-H(K_1,K_2|M)-H(Z^n|M,K_1,K_2)\\
&\stackrel{(a)}{\leq} n(H(Z|U)+I(U;S)+\epsilon_2) - H(U^n|M) - H(Z^n|M,U^n)\\
&= n(H(Z|U)+I(U;S)+\epsilon_2) - H(U^n,Z^n|M)\\
&= n(H(Z|U)+I(U;S)+\epsilon_2) - H(U^n,Z^n|M,S^n) - I(S^n;U^n,Z^n|M)\\
&\stackrel{(b)}{\leq} n(H(Z|U)+I(U;S)+\epsilon_2) - H(Z^n|M,S^n,U^n) - H(S^n|M)+H(S^n|U^n,Z^n)\\
&\stackrel{(c)}{\leq} n(I(S;Z,U)+\epsilon_2) - H(S^n|M)+H(S^n|U^n,Z^n),
\end{align*}
where (a) follows by having $H(Z^n)\leq nH(Z)$ as described above, $H(K_1|Z^n,M)\leq H(K_1)\leq n(I(U;S)-I(U;Z)+2\delta)$, $H(K_2|Z^n,M,K_1)\leq n\epsilon_1$ for some $\epsilon_1\to 0$ as $n\to\infty$ (this is decoding of $K_2$ at eavesdropper having $(Z^n,M,K_1)$, following from steps similar to the proof of Theorem~\ref{thm:Rin1} replacing Bob with Eve, as number of $K_2$ indices is $2^{n(I(U;Z)-\delta)}$, see also Appendix~\ref{App:CoveringIndexDecoding}), $H(K_1,K_2|M)=H(U^n|M)$, having $H(Z^n|M,K_1,K_2)=H(Z^n|M,K_1,K_2,U^n)$ as $U^n(M,K_1,K_2)$ is determined by $(M,K_1,K_2)$, and defining $\epsilon_2=2\delta+\epsilon_1$, (b) is by $H(U^n|M,S^n)\geq 0$, and (c) is same as that of the step (d) of the previous paragraph.

\end{IEEEproof}


\section{Proof of Corollary~\ref{thm:WiretapState}}
\label{App:WiretapState}

From Lemma~\ref{thm:Leakage}, we have
\begin{align*}
I(M;Z^n)\leq nI(S;Z,U)+H(S^n|U^n,Z^n)-H(S^n|M,T)+n\epsilon.
\end{align*}
Here, $H(S^n|U^n,Z^n)=\sum\limits_{i=1}^n H(S_i|S_1^{i-1},U^n,Z^n) \leq \sum\limits_{i=1}^n H(S_i|U_i,Z_i)=nH(S|U,Z)$. In addition, as $S^n$ is generated i.i.d., and also independent of $(M,T)$, we have $H(S^n|M,T)=H(S^n)\geq n(H(S)-\epsilon_1)$ with some $\epsilon_1\to 0$ as $n\to\infty$. (The latter expression follows as we can bound $H(S^n)\geq H(S^n|E)$ where $E$ is an indicator random variable with $E=1$ if $S^n$ is typical. Then, $H(S^n|E)= \Pr\{E=0\}H(S^n|E=0)+\Pr\{E=1\}H(S^n|E=1) \geq \Pr\{E=1\}H(S^n|E=1) \geq (1-\epsilon_0)n(H(S)-\epsilon_0)$ for some arbitrarily small $\epsilon_0$, from which the assertion follows by taking $\epsilon_1=\epsilon_0(1+H(S)-\epsilon_0)$.) Then, using these two observations in the equation above, we obtain $I(M;Z^n)\leq nI(S;Z,U)+nH(S|U,Z)-nH(S)+n\epsilon_1+n\epsilon=n(\epsilon_1+\epsilon)$, which concludes the proof.


\section{Proof of Proposition~\ref{thm:Rout1}}
\label{App:Rout1}

\begin{IEEEproof}
Define a random variable $Q$ uniform over $\{1,\cdots,n\}$ and independent of everything else. (A standard technique of using $Q$ as a time sharing parameter will be utilized in the following.) Also define $U_i=(S_1^{i-1},Z_{i+1}^n)$. We have the following bounds.
\begin{align}
I(S^n;Y^n) &\leq I(X^n,S^n;Y^n)\nonumber\\
&\overset{(a)}{\leq} \sum\limits_{i=1}^n I(X_i,S_i;Y_i)\nonumber\\
&\overset{(b)}{=} n I(X_Q,S_Q;Y_Q|Q)\nonumber\\
&\overset{(c)}{\leq} nI(U_Q,Q,X_Q,S_Q;Y_Q) \label{eq:Rout1eq1}
\end{align}
where (a) is due to the memoryless channel $p(y_i|x_i,s_i)$ and the fact that conditioning does not increase the entropy, (b) follows from the distribution of $Q$, (c) follows as the added term $I(Q;Y_Q)+I(U_Q;Y_Q|Q,X_Q,S_Q)\geq 0$. In addition,
\begin{align}
I(S^n;Y^n) &\leq H(S^n)\nonumber\\
&= \sum\limits_{i=1}^n H(S_i|S_{1}^{i-1})\nonumber\\
&\overset{(a)}{=} \sum\limits_{i=1}^n H(S_i)\nonumber\\
&= nH(S_Q|Q) \label{eq:Rout1eq2},
\end{align}
where (a) holds as $S^n$ has an i.i.d. distribution. Moreover,
\begin{align}
I(S^n;Z^n) &= H(S^n) - H(S^n|Z^n)\nonumber\\
&= \sum\limits_{i=1}^n H(S_i) - H(S_i|S_1^{i-1},Z^n)\nonumber\\
&\overset{(a)}{\geq} \sum\limits_{i=1}^n H(S_i) - H(S_i|Z_i,S_1^{i-1},Z_{i+1}^n)\nonumber\\
&\overset{(b)}{=} \sum\limits_{i=1}^n H(S_i) - H(S_i|Z_i,U_i)\nonumber\\
&\overset{(c)}{=} n I(S_Q;Z_Q,U_Q|Q) \label{eq:Rout1eq3},
\end{align}
where (a) is due the fact that conditioning does not increase the entropy, (b) follows from the definition of $U_i$, and (c) is the standard time sharing argument. Finally, we have
\begin{align}
0&\leq \sum\limits_{i=1}^n I(Z_{i+1}^n;Z_i)\nonumber\\
&= \sum\limits_{i=1}^n I(S_1^{i-1}, Z_{i+1}^n;Z_i) - I(S_1^{i-1};Z_i|Z_{i+1}^n)\nonumber\\
&\overset{(a)}{=} \sum\limits_{i=1}^n I(S_1^{i-1}, Z_{i+1}^n;Z_i) - I(Z_{i+1}^n;S_i|S_1^{i-1})\nonumber\\
&\overset{(b)}{=} \sum\limits_{i=1}^n I(S_1^{i-1}, Z_{i+1}^n;Z_i) - I(S_1^{i-1},Z_{i+1}^n;S_i)\nonumber\\
&\overset{(c)}{=} \sum\limits_{i=1}^n I(U_i;Z_i) - I(U_i;S_i)\nonumber\\
&= n (I(U_Q;Z_Q|Q) - I(U_Q;S_Q|Q)) \label{eq:Rout1eq4},
\end{align}
where (a) follows by Csisz\'ar's sum lemma~\cite{ElGamal:Network11} (note that this is similar to the converse result of the Gel'fand-Pinsker problem given in \cite{Heegard:Capacity81}, where the indices of $S$ and $Z$ are reversed), (b) is due to $I(S_1^{i-1};S_i)=0$ as $S^n$ has an i.i.d. distribution, (c) follows from definition of $U_i$, (d) is the standard time sharing argument.

Now, define $U=(U_Q,Q)$, $S=S_Q$ independent of $Q$ (note that this argument is also used in~\cite{Merhav:Information07} in the corresponding single-letterization arguments), $X=X_Q$, $Y=Y_Q$, and $Z=Z_Q$. Then, \eqref{eq:Rout1eq1} implies $I(S^n;Y^n)\leq nI(U_Q,Q,X_Q,S_Q;Y_Q) \leq nI(U,X,S;Y)=nI(S,X;Y)$ as $U\to (X,S)\to Y$ due to the memoryless channel $p(y|x,s)$. \eqref{eq:Rout1eq2} reduces to $I(S^n;Y^n)\leq nH(S_Q|Q)=nH(S)$, as $S_Q=S$ is independent of $Q$. \eqref{eq:Rout1eq3} implies $I(S^n;Z^n)\geq nI(S_Q;Z_Q,U_Q,Q)=nI(S;Z,U)$ as $I(S_Q;Z_Q,U_Q|Q)=H(S_Q|Q)-H(S_Q|Z_Q,U_Q,Q)=H(S_Q)-H(S_Q|Z_Q,U_Q,Q)=I(S_Q;Z_Q,U_Q,Q)$ due to the independence of $S_Q=S$ and $Q$. \eqref{eq:Rout1eq4} is $0\leq n(I(U_Q;Z_Q|Q) - I(U_Q;S_Q|Q))$. Consider adding $n(I(Q;Z_Q)-I(Q;S_Q))$ to the right hand side of this inequality. We have, as $S_Q=S$ and $Q$ are independent, $I(Q;S_Q)=0$, and hence $n(I(Q;Z_Q)-I(Q;S_Q))\geq 0$. Then, $0\leq n(I(U_Q;Z_Q|Q) - I(U_Q;S_Q|Q))\leq n(I(U_Q,Q;Z_Q)-I(U_Q,Q;S_Q))=n(I(U;Z)-I(U;S))$, which implies $0\leq I(U;Z)-I(U;S)$ as a necessary condition.

This concludes the proof, as combining the bounds above with the fact that any achievable $(R_a,R_l)$ for the given channel $p(y,z|x,s)$ and $p(s)$ should satisfy $\frac{1}{n} I(S^n;Y^n) \geq R_a - \epsilon$ and $\frac{1}{n} I(S^n;Z^n) \leq R_l + \epsilon$ implies the inequalities stated in the proposition.
\end{IEEEproof}


\section{Proof of Proposition~\ref{thm:Rout2}}
\label{App:Rout2}

\begin{IEEEproof}
Let $\Pc_1$ denote the set of $p(u,x|s)$ satisfying $I(U;Y)\geq I(U;S)$, and denote $\Pc_2$ denote the set of $p(u,x|s)$ satisfying $I(U;Z)\geq I(U;S)$. For the channel $p(y,z|x,s)=p(y|x,s)p(z|y)$, any $p\in \Pc_2$ also satisfies $p\in \Pc_1$. Therefore, using the Proposition~\ref{thm:Rout1}, we have that if $(R_a,R_l)$ is achievable, then $(R_a,R_l)\in\Rc_o^{3}$, where
\begin{eqnarray*}
\Rc_o^{3} = \bigcup\limits_{p(u,x|s)}
\left( R_a, R_l \right)
\end{eqnarray*}
satisfying
\begin{eqnarray*}
R_a &\leq& \min \left\{H(S), I(X,S;Y) \right\}\\
R_l &\geq& I(S;Z,U)\\
0 &\leq& I(U;Y)-I(U;S).
\end{eqnarray*}

It remains to show $R_a-R_l$ bound. We add the following bound to the ones stated above. (Note that, we use the same $U_i$ and $Q$-dependent definitions stated in Appendix~\ref{App:Rout1}, and hence the following bound can be added to the ones stated in Appendix~\ref{App:Rout1}.)
\begin{eqnarray*}
n(R_a-R_l) &=& I(S^n;Y^n) - I(S^n;Z^n) \\
&=& I(X^n,S^n;Y^n) - I(X^n,S^n;Z^n) -
\left(I(X^n;Y^n|S^n) - I(X^n;Z^n|S^n)\right)\\
&\overset{(a)}{\leq}& I(X^n,S^n;Y^n) - I(X^n,S^n;Z^n)\\
&\overset{(b)}{=}& I(X^n,S^n;Y^n|Z^n)\\
&=& \sum\limits_{i=1}^n H(Y_i|Y_1^{i-1},Z^n) - \sum\limits_{i=1}^n H(Y_i|Y_1^{i-1},Z^n,X^n,S^n)\\
&\overset{(c)}{\leq}& \sum\limits_{i=1}^n H(Y_i|Z_i) - \sum\limits_{i=1}^n H(Y_i|Z_i,X_i,S_i) \\
&=& \sum\limits_{i=1}^n I(X_i,S_i;Y_i|Z_{i})\\
&\overset{(d)}{=}& nI(X_Q,S_Q;Y_Q|Z_{Q},Q)\\
&\overset{(e)}{\leq}& nI(U_Q,Q,X_Q,S_Q;Y_Q|Z_{Q})\\
&\overset{(f)}{=}& nI(U,X,S;Y|Z)\\
&\overset{(g)}{=}& nI(X,S;Y|Z),
\end{eqnarray*}
where (a) and (b) are due to the degradedness condition, (c) is due to the fact that conditioning does not increase entropy and the Markov chain $(Y_1^{i-1},Z_1^{i-1},Z_{i+1}^n,S_1^{i-1},S_{i+1}^n,X_1^{i-1},X_{i+1}^n)\to (X_i,S_i,Z_i) \to Y_i$ as $(X_i,S_i)$ generates $Y_i$, (d) follows from the memoryless property of the channel, (d) and (f) are due to definitions given in Appendix~\ref{App:Rout1}, (e) is by adding the non-negative term $I(Q;Y_Q|Z_Q)+I(U_Q;Y_Q|Z_Q,Q,X_Q,S_Q)$, and (g) is due to the Markov chain $U\to (X,S)\to (Y,Z)$ as outputs $(y_i,z_i)$ are generated i.i.d. using $(x_i,s_i)$ over the memoryless channel $p(y,z|x,s)=p(y|x,s)p(z|y)$. The result then follows by taking the union over all joint distributions $p(u,x|s)$.
\end{IEEEproof}


\section{Proof of the converse for Theorem~\ref{thm:RevDegDMC}}
\label{App:RevDegDMC}

\begin{IEEEproof}
We bound the rate difference as follows.
\begin{eqnarray*}
n(R_a-R_l) &\leq& I(S^n;Y^n) - I(S^n;Z^n) \\
&=& \sum\limits_{i=1}^n I(S^n;Y_i|Y_1^{i-1})
- I(S^n;Z_i|Z_{i+1}^n) \\
&\overset{(a)}{=}& \sum\limits_{i=1}^n
I(S^n;Y_i|Y_1^{i-1},Z_{i+1}^n)
+ I(Z_{i+1}^n;Y_i|Y_1^{i-1})
- I(Z_{i+1}^n;Y_i|Y_1^{i-1},S^n)\nonumber\\
&&\:{-} \left[
I(S^n;Z_i|Y_1^{i-1},Z_{i+1}^n)
+ I(Y_{1}^{i-1};Z_i|Z_{i+1}^{n})
- I(Y_{1}^{i-1};Z_i|Z_{i+1}^{n},S^n)
\right] \\
&\overset{(b)}{=}& \sum\limits_{i=1}^n
I(S^n;Y_i|Y_1^{i-1},Z_{i+1}^n)
-I(S^n;Z_i|Y_1^{i-1},Z_{i+1}^n)\\
&=&  \sum\limits_{i=1}^n I(S_i;Y_i|Y_1^{i-1},Z_{i+1}^n,S_{1}^{i-1},S_{i+1}^n)
+I(S_{1}^{i-1},S_{i+1}^n;Y_i|Y_1^{i-1},Z_{i+1}^n) \nonumber\\
&&{}\:{-}I(S_i;Z_i|Y_1^{i-1},Z_{i+1}^n,S_{1}^{i-1},S_{i+1}^n)
-I(S_{1}^{i-1},S_{i+1}^n;Z_i|Y_1^{i-1},Z_{i+1}^n)\\
&\overset{(c)}{\leq}& \sum\limits_{i=1}^n
I(S_i;Y_i|U_i) - I(S_i;Z_i|U_i)\\
&\overset{(d)}{=}& n[I(S;Y|U)-I(S;Z|U)],
\end{eqnarray*}
where, in (a),
we used the equalities
\begin{equation*}
I(S^n;Y_i|Y_1^{i-1}) + I(Z_{i+1}^n;Y_i|Y_1^{i-1},S^n)
=I(Z_{i+1}^n;Y_i|Y_1^{i-1}) + I(S^n;Y_i|Y_1^{i-1},Z_{i+1}^n)
\end{equation*}
and
\begin{equation*}
I(S^n;Z_i|Z_{i+1}^n) + I(Y_{1}^{i-1};Z_i|Z_{i+1}^{n},S^n)
=I(Y_{1}^{i-1};Z_i|Z_{i+1}^{n}) + I(S^n;Z_i|Z_{i+1}^{n},Y_{1}^{i-1});
\end{equation*}
in (b), we used the Csiszar's sum lemma~\cite{Csisz'ar:Broadcast78}
(and a conditional form of it) to obtain the equalities
\begin{equation*}
\sum\limits_{i=1}^n I(Z_{i+1}^n;Y_i|Y_1^{i-1})
= \sum\limits_{i=1}^n I(Y_{1}^{i-1};Z_i|Z_{i+1}^{n})
\end{equation*}
and
\begin{equation*}
\sum\limits_{i=1}^n I(Z_{i+1}^n;Y_i|Y_1^{i-1},S^n)
= \sum\limits_{i=1}^n I(Y_{1}^{i-1};Z_i|Z_{i+1}^{n},S^n);
\end{equation*}
in (c), we define
$U_i\triangleq (Y_1^{i-1},Z_{i+1}^n,S_{1}^{i-1},S_{i+1}^n)$;
and use the reversely degradedness of the channel resulting in
$I(S_{1}^{i-1},S_{i+1}^n;Z_i|Y_1^{i-1},Z_{i+1}^n)\geq I(S_{1}^{i-1},S_{i+1}^n;Y_i|Y_1^{i-1},Z_{i+1}^n)$, in (d), we obtain the single-letter expression
(by defining $U=(U_i,i)$, etc., see, e.g., \cite{Cover:Elements91}).
Note that, with the definition of $U_i$ in (c),
$X_i$ is generated from $S^n$, and hence from $(U_i,S_i)$.
In addition, given $(X_i,S_i)$, $Z_i$ is independent of $(U_i,S_i)$.
Thus, $(U,S)\to (X,S)\to Z$ forms a Markov chain, and the
upper bound is given by
\begin{eqnarray*}
R_a-R_l &\leq& \max\limits_{p(u,x|s) \: \textrm{s.t.} \: (U,S)\to (X,S)\to Z}
I(S;Y|U) - I(S;Z|U)\\
&=& \max\limits_{p(x|u^*,s),\: u^*\in\Uc} I(S;Y|U=u^*)-I(S;Z|U=u^*)\\
&=& \max\limits_{p(x|s)} I(S;Y)-I(S;Z),
\end{eqnarray*}
where the equalities follow due to the following:
First, the conditional mutual information expression is
maximized with a particular input $u^*$ (as randomizing
over different $u$ values will not increase the sum
$\sum\limits_{u} (I(S;Y|U=u) - I(S;Z|U=u))\Prob\{U=u\})$)
and a probability distribution $p^*(x|u^*,s)$.
Second, the optimal $p^*(x|u^*,s)$ will correspond to a
$p(x|s)$. Thus, the converse result can be stated over
input distributions in the form $p(x|s)$, matching to the achievability result.
\end{IEEEproof}





\end{document}